\documentclass[
reqno,
11pt,
oneside
]{article}

\usepackage[ngerman, english]{babel}
\usepackage[utf8]{inputenc}
\usepackage{enumerate}
\usepackage[normalem]{ulem}
\usepackage[babel,french=guillemets,german=swiss]{csquotes}
\usepackage[center]{caption}
\usepackage{cite}
\usepackage{enumitem}
\usepackage{mathtools}

\usepackage[%
left=1.25in,
right=1.25in,
bottom=1.25in,
top=1in,
footskip=0.5in,
]{geometry}

\usepackage{color}
\definecolor{LinkColor}{rgb}{0,0,1}
\definecolor{LinkColor2}{rgb}{0,0.5,0}
\definecolor{lbcolor}{rgb}{0.85,0.85,0.85}
\definecolor{FrameColor}{rgb}{0.85,0.85,0.85}

\usepackage{lineno}

\newcommand*\patchAmsMathEnvironmentForLineno[1]{%
	\expandafter\let\csname old#1\expandafter\endcsname\csname #1\endcsname
	\expandafter\let\csname oldend#1\expandafter\endcsname\csname end#1\endcsname
	\renewenvironment{#1}%
	{\linenomath\csname old#1\endcsname}%
	{\csname oldend#1\endcsname\endlinenomath}}%
\newcommand*\patchBothAmsMathEnvironmentsForLineno[1]{%
	\patchAmsMathEnvironmentForLineno{#1}%
	\patchAmsMathEnvironmentForLineno{#1*}}%
\AtBeginDocument{%
	\patchBothAmsMathEnvironmentsForLineno{equation}%
	\patchBothAmsMathEnvironmentsForLineno{align}%
	\patchBothAmsMathEnvironmentsForLineno{flalign}%
	\patchBothAmsMathEnvironmentsForLineno{alignat}%
	\patchBothAmsMathEnvironmentsForLineno{gather}%
	\patchBothAmsMathEnvironmentsForLineno{multline}%
}

\usepackage{enumitem}

\usepackage{graphicx}
\usepackage{placeins}

\usepackage{amsmath}
\usepackage{amssymb}
\usepackage{dsfont}
\usepackage{empheq}
\usepackage{amsthm}
\numberwithin{equation}{section}

\usepackage[%
pdftitle={Titel},%
pdfauthor={Autor},%
pdfcreator={LaTeX, LaTeX with hyperref and KOMA-Script},
pdfsubject={Betreff}, 
pdfkeywords={Keywords}
]{hyperref} 

\hypersetup{%
	colorlinks	=true,
	linkcolor	=LinkColor,%
	anchorcolor	=LinkColor,%
	citecolor	=LinkColor2,%
	filecolor	=LinkColor,%
	menucolor	=LinkColor,%
	urlcolor	=LinkColor,%
}

\newtheorem{theorem}{Theorem}[section]
\newtheorem{lemma}[theorem]{Lemma}

\newtheorem{definition}[theorem]{Definition}

\theoremstyle{definition}
\newtheorem{remark}[theorem]{Remark}

\makeatletter
\renewenvironment{proof}[1][\proofname]{%
	\par\pushQED{\qed}\normalfont%
	\topsep6\p@\@plus6\p@\relax
	\trivlist\item[\hskip\labelsep\bfseries#1\@addpunct{.}]%
	\ignorespaces
}{%
	\popQED\endtrivlist\@endpefalse
}
\makeatother

\makeatletter
\renewcommand\paragraph{\@startsection{paragraph}{4}{\z@}%
	{1ex \@plus1ex \@minus.2ex}%
	{-1em}%
	{\normalfont\normalsize\bfseries}}
\renewcommand\subparagraph{\@startsection{paragraph}{4}{\z@}%
	{1ex \@plus1ex \@minus.2ex}%
	{-1em}%
	{\normalfont\normalsize\itshape}}
\makeatother

\newcommand{\norm}[1]{\ensuremath\left\| #1 \right\|}
\newcommand{\bignorm}[1]{\ensuremath\big\| #1 \big\|}
\newcommand{\abs}[1]{\ensuremath\left|#1 \right|}

\def\R{\mathbb R}

\def\eps{\varepsilon}
\def\supp{\mathrm{supp\,}}

\DeclareMathOperator{\curl}{curl}
\DeclareMathOperator{\divg}{div}

\def\dint{\;\mathrm d}
\def\dtau{\;\mathrm d\tau}
\def\dx{\;\mathrm dx}
\def\dv{\;\mathrm dp}
\def\dz{\;\mathrm dz}

\def\dy{\;\mathrm dy}

\def\ds{\;\mathrm ds}

\def\d{{\mathrm{d}}}

\def\del{\partial}
\def\delt{\partial_{t}}
\def\delx{\partial_{x}}
\def\delp{\partial_{p}}
\def\delr{\partial_{r}}

\def\E{\mathcal E}
\def\F{\mathcal F}
\def\G{\mathcal G}
\def\C{\mathcal{C}}
\def\H{\mathcal{H}}
\def\RR{\mathcal{R}}
\def\SS{\mathcal{S}}

\def\ov{\overline}

\newcommand{\suchthat}{\;\ifnum\currentgrouptype=16 \middle\fi|\;}

\allowdisplaybreaks

\font\myfont=cmr8 at 15pt

\begin{document}
	
%
%

\title{\myfont
	On the two and one-half dimensional Vlasov--Poisson system\\ 
	with an external magnetic field:\\
	Global well-posedness and stability of confined steady states}

\bigskip

\author{Patrik Knopf \footnotemark[1] \and Jörg Weber \footnotemark[2]}

\date{\today}

\renewcommand{\thefootnote}{\fnsymbol{footnote}}
\footnotetext[1]{Department of Mathematics, University of Regensburg, 93053 Regensburg, Germany \\
	\tt(\href{mailto:patrik.knopf@ur.de}{patrik.knopf@ur.de})}
\footnotetext[2]{Centre of Mathematical Sciences, Lund University, 221 00 Lund, Sweden \\
	\tt(\href{mailto:jorg.weber@math.lu.se}{jorg.weber@math.lu.se})}

\maketitle

\begin{center}
	\small
	{
		\textit{This is a preprint version of the paper. Please cite as:} \\  
		P.~Knopf and J.~Weber, Nonlinear Anal. Real World Appl. 65:103460 (2022) \\ 
		\url{https://doi.org/10.1016/j.nonrwa.2021.103460}
	}
	\bigskip
\end{center}

%
%

\begin{abstract}
The time evolution of a two-component collisionless plasma is modeled by the Vlasov--Poisson system. In this work, the setting is two and one-half dimensional, that is, the distribution functions of the particles species are independent of the third space dimension. We consider the case that an external magnetic field is present in order to confine the plasma in a given infinitely long cylinder. After discussing global well-posedness of the corresponding Cauchy problem, we construct stationary solutions whose support stays away from the wall of the confinement device. Then, in the main part of this work we investigate the stability of such steady states, both with respect to perturbations of the initial data, where we employ the energy-Casimir method, and also with respect to perturbations of the external magnetic field.
\\[1ex]
\end{abstract}

\setlength\parindent{0ex}
\setlength\parskip{1ex}

\bigskip


%
%

\section{Introduction}

The investigation of a high-temperature plasma that is influenced by an exterior magnetic field is one of the most important aspects in the research on thermonuclear fusion. In this context, the \emph{magnetic confinement} of plasmas is of particular interest. The idea of magnetic confinement is to apply magnetic fields to influence the plasma in such a way that it keeps a certain distance to the reactor wall. This is necessary because a plasma colliding with the wall will cool down (which interrupts the fusion process) and might even damage the reactor. 
The most common reactor types for magnetic confinement are toroidal devices, for instance, Tokamaks, Stellerators and reversed field pinches. However, there are also linear (usually cylindrical) confinement devices, such as $z$-pinch or $\theta$-pinch configurations. Their advantage is that they can be described more easily by mathematical models. For more details on thermonuclear fusion from a physics point of view we refer to \cite{Stacey}.

In this paper, we intend to study the behavior of a plasma in a long cylindrical reactor. As a simplification for the mathematical analysis, we assume that the length of this cylinder is infinite. Moreover, we suppose that the volume charge density of the plasma is constant along lines parallel to the cylinder axis. 

There are different possibilities of describing the time evolution of plasmas by means of mathematical equations. A very popular ansatz is to regard the plasma as an electrically conducting fluid. Its behavior is then described by the \emph{magneto-hydrodynamic (MHD) equations}, which are a combination of the Navier–Stokes equations of fluid dynamics and Maxwell’s equations of electro-magnetism.

Another way to describe the motion of a plasma are kinetic equations. These models are particularly suitable if collisions among the ions and electrons can be neglected. It has been observed that such collisions only play a minor role in very hot and sufficiently diluted plasmas. The motion of a plasma can then be described by a \emph{Vlasov equation} which will be the model of choice in our work.

To describe the time evolution of a two-component plasma (consisting of both electrons and positively charged ions) in a cylinder of infinite length, we consider the following Vlasov--Poisson system that is equipped with an external magnetic field $B=\curl_x A$:
\begin{subequations}\label{eq:WholeSystem}
	\begin{alignat}{2}
	&\delt f^\pm + p\cdot\delx f^\pm \pm (-\delx U+p\times B)\cdot\delp f^\pm =0 
	&&\quad\mathrm{on}\ \R_0^+ \times\R^2\times\R^3,
	\\[1ex]
	&-\Delta_x U=4\pi \rho, 
	&&\quad\mathrm{on}\ \R_0^+ \times\R^2,\label{eq:Poisson}\\[1ex]
	&\lim_{|x|\to\infty}(U(t,x)+2M\ln|x|)=0 
	&&\quad\mathrm{on}\ \R_0^+ ,\label{eq:PoissonBC}\\[1ex]
	&\rho^\pm = \int_{\R^3} f^\pm(\cdot,\cdot,p) \, \mathrm dp, 
	    \quad \rho = \rho^+ - \rho^-
	&&\quad\mathrm{on}\ \R_0^+ \times\R^2,\label{eq:Rho}\\[1ex]
	&B = \curl_x A 
	&&\quad\mathrm{on}\ \R_0^+ \times\R^2,\label{eq:Mag}\\[1ex]
	&f^\pm(0)=\mathring f^\pm 
	&&\quad\mathrm{on}\ \R^2\times\R^3.\label{eq:Initf}
	\end{alignat}
\end{subequations}
Here, the plasma is represented by the electron density $f^-=f^-(t,x,p) \ge 0$ and the ion density $f^+=f^+(t,x,p) \ge 0$, which both depend on time $t\in\R_0^+:=[0,\infty[$, position $x\in\R^2$ and velocity $p\in\R^3$. In the standard three-dimensional Vlasov--Poisson system, the position coordinate $x$ is also a vector in $\R^3$. However, as we consider the plasma in a cylinder of infinite length, we assume that the density function is constant along the direction of the cylinder axis, which in turn is supposed to point in the $x_3$ direction. This means that the distribution function $f$ does not depend on the $x_3$ variable and it thus suffices to consider $x=(x_1,x_2)\in\R^2$.
Since $f$ may still depend on velocity vectors in $\R^3$, the system \eqref{eq:WholeSystem} is referred to as the \emph{two and one-half dimensional Vlasov--Poisson system}. The initial condition for $f$ is given in \eqref{eq:Initf}.
In the following, we write $z=x_3$ to denote the third spatial component and thus, the $x_3$-axis will also be referred to as the $z$-axis.

The scalar function $U$ denotes the electrostatic potential which is induced by the charge of the particles. We point out that in the Vlasov--Poisson system, electromagnetic effects such as the induction of an internal magnetic field are neglected.
Hence, the field equation for $U$ is the Poisson equation \eqref{eq:Poisson} with the charge density $\rho = \rho^+ - \rho^-$, given by \eqref{eq:Rho}, as its right-hand side. 
The internal self-consistent electric field, which drives the movement of the charged particles, is then given as $E=-\delx U$.
In order to determine $U$ uniquely for given $\rho$, we impose the boundary condition \eqref{eq:PoissonBC} at spatial infinity. Here, $M=\int_{\R^2}\rho\dx$ is the total charge in any cross section of the cylinder. Note that $M$ is a conserved quantity, i.e., it does not depend on time. We point out that in our scenario, the electric potential $U$ does not vanish at spatial infinity as it would be the case in a fully three-dimensional setting. This is because the fundamental solution $-\frac{1}{2\pi}\ln|\cdot|$ of Poisson's equation in $2D$ does not vanish at infinity. For a more detailed discussion on this issue, we refer to \cite{knopfStSt}. Throughout this paper we use modified Gaussian units such that all physical constants (mass and modulus of the charge of a particle) are normalized to unity.

In \cite{knopfStSt}, confined steady states for the two-dimensional analogue of \eqref{eq:WholeSystem} (meaning that ${x,p\in\R^2}$) were constructed. The same has been done in \cite{weberStSt} for the two and one-half dimensional relativistic Vlasov--Maxwell system. In both papers an additional symmetry, namely rotational symmetry about the $z$-axis, is assumed, and the ansatz $$f_0^\pm=\eta^\pm(\E^\pm,\F^\pm,\G^\pm)$$ 
for a steady state is made, where $\E^\pm$, $\F^\pm$, and $\G^\pm$ are invariants of the characteristic flow. For the steady states in our setting, the particle energy $\E^\pm=\frac12|p|^2\pm U$, which is (for the variables $(x,p\pm A)$) the Hamiltonian governing the motion of the particles (ions or electrons, respectively), is the invariant corresponding to time symmetry. Moreover, from the Lagrangian $$\mathcal L^\pm=-\sqrt{1-|(\dot x,\dot z)|^2}\mp U\pm (\dot x,\dot z)\cdot A \quad\text{on $\R^3\times\R^3$},$$ 
we derive the invariant quantities corresponding to the spatial symmetries: the canonical angular momentum $\F^\pm=\partial_{\dot\varphi}\mathcal L^\pm=r(p_\varphi\pm A_\varphi)$ corresponds to rotational invariance, and $\G^\pm=\partial_{\dot z}\mathcal L=p_3\pm A_3$, that is the third component of the canonical momentum, corresponds to translational invariance. Here and throughout this paper, we frequently use cylindrical coordinates $(r,\varphi,z)$ on $\R^3$. In this context, $v_r$, $v_\varphi$, and $v_3$ denote the coordinates of a vector $v$ in the corresponding local coordinate system. 

In this paper we study the stability properties of steady states of the separated form 
\begin{align*}
    f_0^\pm=\vartheta^\pm(\E^\pm)\, \psi^\pm(\F^\pm,\G^\pm)
\end{align*}
with respect to both perturbations of the initial data and perturbations of the external magnetic field. As for stability with respect to perturbations of the initial data, it is well-known in the study of Vlasov type systems (both in the gravitational and plasma physics case) that a certain monotonicity assumption on $\vartheta^\pm$ is important; see \ref{COND:S4}.

\paragraph{Comparison to related results in the literature.}
The mathematical investigation of stationary solutions in kinetic theory and their stability has an extensive history. 
A detailed study of steady states of the three-dimensional Vlasov--Poisson system, both in the gravitational case (without any external field) and in the plasma physics case (where no confinement properties are taken into account), can be found in \cite{ReinBook} and the large amount of references therein. In particular, we refer to \cite{GuoRein,LMR,rein-steady,Rein_Casimir,Rein_Casimir_2,rein02}. Within this family of works, we especially want to point out \cite{BRV_Stability}, which is closest to our strategy.
For the mathematical investigation of steady states in similar Vlasov type models, we refer to \cite{guo-grotta,duan,duan2,glassey,batt2} to mention but a few.

In order to investigate steady states of plasmas described by Vlasov type models, the system has to be equipped with a suitable external field. Besides the aforementioned papers \cite{knopfStSt} and \cite{weberStSt}, confined steady states in various settings under the influence of an external magnetic field were constructed in \cite{belyaeva2,belyaeva3,skubachevskii}. 
In the non-stationary case, an external magnetic field becoming infinite at the boundary of the confinement device was used in \cite{CCM12,CCM14,CCM15,CCM16} to confine a Vlasov--Poisson plasma. More information on the confinement problem in various configurations with different Vlasov type plasma models (in some cases addressing linear stability), can be found in \cite{han,NNS,NS,Zhang,ZHR,dolbeault}. 
We further want to mention the paper \cite{BGT} dealing with instability of $z$-pinches in an MHD model.

\paragraph{Structure of this paper.}
After stating some important auxiliary results in Section~\ref{sec:prel}, we discuss the global-in-time classical well-posedness of system \eqref{eq:WholeSystem} in Section~\ref{sec:global-well-posed}. Here, we employ some of the strategies that have already been used in the literature (see \cite{batt,kurth,lions-perthame,pfaffelmoser,schaeffer} for results on local and global well-posedness) to tackle the three-dimensional Vlasov--Poisson system.
In Section~\ref{sec:Existence}, we briefly outline the construction of confined steady states in our setting proceeding similarly to \cite{knopfStSt,weberStSt}. The main novelties of the present paper are the stability results for such confined steady states, which are presented in Section~\ref{sec:stability}. Restricting ourselves to a $\theta$-pinch configuration (i.e., the external magnetic field points in the direction of the cylinder axis), we first prove stability of a previously constructed confined steady state with respect to perturbations of the initial data (see Theorem \ref{thm:Stability}). This is done by suitably adapting the strategy of \cite{BRV_Stability}. Furthermore, for perturbations of the external magnetic field, we also obtain a stability estimate (see Theorem~\ref{thm:Stability_B}). However, this estimate has the drawback of not being uniform in time. Eventually both results can be combined to obtain a stability estimate with respect to perturbations of both the initial data and the external magnetic field (see Theorem~\ref{thm:Stability_C}).














%
%

\section{Preliminaries}\label{sec:prel}



First, we recall the following classical result on the Newtonian potential in two dimensions (see, e.g., \cite[Theorem~10.2 and Theorem~10.3]{lieb-loss}).
\begin{lemma} \label{LEM:EP}
    Let $\varrho \in C^1_c(\R^2)$ be arbitrary. We set 
    $ M_\varrho := \int_{\R^2}\varrho\dx$.
    
    Then the function
    \begin{align}
    \label{DEF:URHO}
        U_\varrho: \R^2\to\R, \quad
        U_\varrho(x) := - 2 \int_{\R^2} \varrho(y) \ln\abs{x-y} \dy
    \end{align}
    belongs to $C^2(\R^2)$ and is the unique classical solution of Poisson's equation
    \begin{gather*}
	-\Delta_x U_\varrho =4\pi \varrho
	\quad\mathrm{on}\ \R^2,
	\\[1ex]
	\lim_{|x|\to\infty}(U_\varrho(t,x)+2M_\varrho \ln|x|)=0. 
	\end{gather*}
	The gradient of $U_\varrho$ is given as
	\begin{align*}
        \nabla U_\varrho(x) = - 2 \int_{\R^2} \varrho(y) \frac{x-y}{\abs{x-y}^2} \dy,
        \quad x\in\R^2.
    \end{align*}
    In particular, this entails that $U_\varrho(x) = \mathcal O(\abs{x}^{-1})$ as $\abs{x}\to\infty$.
	Moreover, we obtain the estimates
	\begin{align}
	    \label{EST:DU}
	    \norm{\nabla U_\varrho}_{\infty} 
	    &\le c_1 \norm{\varrho}_{q}^{q/2}
	    \norm{\varrho}_{\infty}^{1-q/2} 
	    \quad\text{for all $q\in[1,2[$,}\\
	    \label{EST:D2U}
	    \norm{D^2 U_\varrho}_{\infty}
	    &\le c_2 \Big[ \big(1+\norm{\varrho}_{\infty}\big) \Big(1+\ln_+\big(\norm{\nabla\varrho}_{\infty}\big)\Big) + \norm{\varrho}_{1}\Big]
	\end{align}
	with positive constants $c_1=c_1(q)$ and $c_2$ that do not depend on $\varrho$.
\end{lemma}

Next, we state two important properties of the potential energy: an upper bound is provided by the logarithmic Hardy--Littlewood--Sobolev inequality (see \cite[Theorem~1]{CarlenLossHLS}), and positive definiteness under the assumption that the total charge is zero is ensured by \cite[Theorem~1.16]{Landkof}. 

\begin{lemma}\label{lem:estimates_potential_energy}
Suppose that $\rho\in L^1(\R^2)$.
\begin{enumerate}[label= $\mathrm{(\alph*)}$]
    \item If $\rho$ is non-negative almost everywhere on $\R^2$, and the integral $$\int_{\R^2}\rho(x)\ln(1+|x|^2)\dx$$ is finite,
    then there exists a constant $C>0$ independent of $\rho$ such that
    \[-\int_{\R^2}\int_{\R^2}\ln|x-y|\rho(y)\rho(x)\dy\dx\le\frac12\|\rho\|_1\int_{\R^2}\rho\ln\frac{\rho}{\|\rho\|_1}\dx+C\|\rho\|_1^2.\]
    \item If $\int_{\R^2}\rho\dx=0$, it holds that
    \[\int_{\R^2}\int_{\R^2} \big(-\ln|x-y|\big) \rho(y)\rho(x)\dy\dx\ge0,\]
    provided that this integral exists.
\end{enumerate}
\end{lemma}



%
%

\section{Global classical well-posedness}\label{sec:global-well-posed}

We first want to establish the existence of a unique global classical solution $f=(f^+,f^-)$ with $f^\pm\in C^1([0,\infty[;\R^5)$ to the system \eqref{eq:WholeSystem}.

To this end,
we first recall that for any functions $F,G\in C_b([0,\infty[;C^1_b(\R^2))$,
$f$ is a classical solution to the equations
\begin{align}
\label{GEN}
 \delt f^\pm + p\cdot\delx f^\pm \pm (F+p\times G)\cdot\delp f^\pm =0 
 \quad\text{on $\R_0^+\times\R^5$}
\end{align}
if and only if the functions $f^\pm$ are constant along solutions of the characteristic systems
\begin{align}
\label{CHARSYS}
    \dot x = p,\quad
    \dot p = \pm \big( F(\cdot,x) + p \times G(\cdot,x) \big).
\end{align}
This equivalence is a direct consequence of \cite[Lemma~4]{knopf} and \cite[Lemma~5]{knopf} (which are stated in the three-dimensional setting but can easily be transferred to the two-and-one-half dimensional case).
We directly infer that any classical solution $f$ to \eqref{GEN} can be expressed as
\begin{align*}
    f^\pm(t,z) = \mathring f^\pm\big(Z^\pm(0,t,z)\big)
    \quad\text{for all $t\ge 0$, $z=(x,p)\in\R^5$,}
\end{align*}
where $Z^\pm(\cdot,t,z)=(X^\pm,P^\pm)(\cdot,t,z)$ denotes the unique solution of \eqref{CHARSYS} satisfying the condition $Z^\pm(t,t,z) = z$ for any given $z=(x,p)\in\R^5$. One can easily see that the characteristic flow is measure preserving, meaning that
\begin{align*}
    \det\left(\frac{\partial Z}{\partial z}\right) \equiv 1.
\end{align*}
For more details, we refer to \cite[Lemma~4]{knopf}. The above results directly imply that any classical solution $f=(f^+,f^-)$ to \eqref{GEN} satisfies
\begin{align*}
    \supp{f^\pm(t)} = Z^\pm(t,0,\supp{\mathring f^\pm})
    \quad\text{and}\quad
    \bignorm{f^\pm(t)}_{L^q(\R^5)} 
    = \bignorm{\mathring f^\pm}_{L^q(\R^5)}
\end{align*}
for all $q\in[1,\infty]$ and all $t\ge 0$ for which the solution exists (cf.~\cite[Lemma~5]{knopf}).

For any classical solution $f$ to the system \eqref{eq:WholeSystem}, we further define the quantities
\begin{align*}
    \mathbb P^\pm(t) &:= \sup \big\{ |P^\pm(s,0,z)| \;\big|\; z \in \supp{\mathring f^\pm},\; 0\le s\le t\,\big\}\\
    &\phantom{:}= \sup \big\{ \abs{p} \;\big|\; 
    \exists s\in[0,t]:\; (x,p) \in \supp{f^\pm(s)} \big\},\\[1ex]
    \mathbb X^\pm(t) &:= \sup \big\{ |X^\pm(s,0,z)| \;\big|\; z \in \supp{\mathring f^\pm},\; 0\le s\le t\,\big\}\\
    &\phantom{:}= \sup \big\{ \abs{x} \;\big|\; 
    \exists s\in[0,t]:\; (x,p) \in \supp{f^\pm(s)} \big\},\\[1ex]
    \mathbb P(t) &:= \underset{\pm}{\max}\; \mathbb P^\pm(t),
    \quad \mathbb X(t) := \underset{\pm}{\max}\; \mathbb X^\pm(t),
\end{align*}
for all $t\ge 0$ at which the solution exists. Here, $(X^\pm,P^\pm)$ denotes the solutions of the characteristic systems \eqref{CHARSYS} for the choices $G=B$ and $F=-\delx U$, where $U$ is the self-consistent electric potential satisfying \eqref{eq:Poisson}--\eqref{eq:Rho} and $B$ is given by \eqref{eq:Mag}.
We point out that the functions $\mathbb P^\pm$, $\mathbb X^\pm$, $\mathbb P$ and $\mathbb X$ are non-decreasing but we cannot a priori exclude the possibility that they become infinite already in finite time. However, to show that the classical solutions exist globally in time, it will be crucial to show that $\mathbb P(t)$ actually remains bounded on finite time intervals.

Furthermore, we consider the total energy
\begin{align*}
    \H\big(f(t)\big) 
    = E_\mathrm{kin}\big(f(t)\big) + E_\mathrm{pot}\big(f(t)\big),
\end{align*}
where the kinetic energy $E_\mathrm{kin}$ and the potential energy $E_\mathrm{pot}$ are defined as
\begin{align*}
    E_\mathrm{kin}\big(f(t)\big) &:= \frac 12 \sum_\pm \int_{\R^2}\int_{\R^3}\abs{p}^2 f^\pm(t,x,p) \dv\dx,\\
    E_\mathrm{pot}\big(f(t)\big) &:= \frac 12 \int_{\R^2} U(t,x) \rho(t,x) \dx. 
\end{align*}
It is straightforward to check that the total energy $\H\big(f(t)\big) $
is constant in time, meaning that $\H(f(t)) = \H(\mathring f)$, as long as the solution exists.

The global well-posedness result reads as follows:

\begin{theorem}
    \label{thm:globalWP}
    Let $\mathring f = (\mathring f^+,\mathring f^-) $ with $\mathring f^\pm \in C^1_c(\R^5;\R_0^+)$ be arbitrary initial data, and let $A\in C_b([0,\infty[;C^2_b(\R^2;\R^3))$ be any given external magnetic vector potential. 
    
    Then there exists a unique classical solution $f=(f^+,f^-)$ with ${f^\pm\in C^1([0,\infty[\times\R^5)}$ such that the supports $\supp f^\pm \subset \R^5$ remain compact for all $t\ge 0$. 
    
    Moreover, for any $\eps>0$, there exists a constant $c>0$ depending only on $\eps$, $\H\big(\mathring f\big)$, $\bignorm{\mathring f^\pm}_1$ and $\bignorm{\mathring f^\pm}_\infty$ such that
    \begin{align}
    \label{EST:P}
        \mathbb P(t) \le c(1+t)^{2+\eps}
        \quad\text{for all $t \ge 0$.}
    \end{align}
\end{theorem}

\begin{proof}
    \textit{Step 1: Local well-posedness.} 
    For the standard three-dimensional Vlasov--Poisson system (i.e., without any external magnetic field), the existence and uniqueness of a local classical solution was first established by R.~Kurth~\cite{kurth}. 
    As the external magnetic field 
    ${B=\curl_x A \in C_b([0,\infty[;C^1_b(\R^2;\R^3))}$ 
    is sufficiently regular, and the electric potential can be bounded by means of the estimates \eqref{EST:DU} and \eqref{EST:D2U}, we can proceed completely analogously to Kurth's proof to establish the existence of a local classical solution. A sketch of the proof for the corresponding local existence result in three dimensions (with external magnetic field) can be found in \cite[Theorem~1]{knopf}.
    In this way, we obtain a unique local classical solution $f=(f^+,f^-)$ to the system \eqref{eq:WholeSystem} existing on a right-maximal time interval $[0,T_*[$. 
    
    \textit{Step 2: Continuation criterion.}
    For the local classical solution of the standard three-dimensional Vlasov--Poisson system (i.e., $B\equiv 0$), a continuation criterion was established by J.~Batt~\cite{batt}. It states that the local solution can be extended in time as long as the velocity support $\mathbb P(t)$ remains under control.
    Proceeding analogously, we conclude that for the maximal local classical solution $f=(f^+,f^-)$ constructed in Step~1, it either holds that
    \begin{align}
    \label{CONTCRIT}
        \underset{\substack{t\to T_* \\ 0\le t <T_*}}{\lim} \mathbb P(t) = \infty
    \end{align}
    or $T_*=\infty$. Hence, in order to show that our unique maximal solution actually exists globally in time, we must ensure that $\mathbb P(t)$ remains bounded on finite time intervals as long as the solution exists. 
    
    \textit{Step 3: Extension on $[0,\infty[$.}
    Based on Batt's continuation criterion, two different proofs for global existence of classical solutions
    for the standard three-dimensional Vlasov--Poisson system (i.e., $B\equiv 0$) were virtually obtained simultaneously: one of them by P.L.~Lions and B.~Perthame~\cite{lions-perthame}, and the other one by K.~Pfaffelmoser~\cite{pfaffelmoser}. Pfaffelmoser's proof was later greatly simplified by J.~Schaeffer~\cite{schaeffer}. 
    We point out that the Pfaffelmoser--Schaeffer proof was further generalized in \cite[Theorem~1]{knopf} to handle the three-dimensional Vlasov--Poisson system endowed with an external magnetic field. 
    
    Although our very basic strategy is the same, we will see that for or our two-and-one-half dimensional system, the proof of global existence differs greatly from the three-dimensional case. Fortunately, we have much better estimates for the electric field $-\delx U$ (see Lemma~\ref{LEM:EP}) than it would be the case in three dimensions. This means that the main difficulties of the proofs by Lions and Perthame or Pfaffelmoser and Schaeffer do not occur in our situation. Therefore,  our proof will be much more straightforward. 
    However, we will face different problems as the two-dimenional
    electric potential is given by a logarithmic convolution kernel and thus, it is non-trivial to establish that the potential energy is bounded from below.
    
    We argue by contradiction and assume that $T_*<\infty$. Hence, the goal is to show that $\mathbb P(t)$ remains bounded for all $t\in[0,T_*[$. This would be a contradiction to \eqref{CONTCRIT} and we thus conclude that actually $T_*=\infty$.
    To this end, let $t\in[0,T_*[$ be arbitrary. In the following, the letter $C$ will denote generic positive constants that depend only $\bignorm{\mathring f^\pm}_1$ and $\bignorm{\mathring f^\pm}_\infty$, and may change their value from line to line.
    
    It is crucial that an external magnetic field does not change the modulus of a particle's velocity, but only its direction. For any $s\in [0,t]$, we thus obtain the estimate
    \begin{align*}
        \frac12|P^\pm(s)|^2 
        \le \frac12|P^\pm(0)|^2 
        + \int_0^s \abs{\delx U\big(\tau,X(\tau)\big)}\, 
        \abs{P^\pm(\tau)} \dtau 
    \end{align*}
    for all $s\in[0,t]$. Invoking a quadratic version of Gronwall's lemma (see \cite[Theorem~5]{Dragomir}), and taking the maximum of the components $\pm$, we conclude that
    \begin{align}
    \label{EST:PS}
        \mathbb P(s) 
        \le \mathbb P(0)  
        + \int_0^s \norm{\delx U\big(\tau\big)}_{\infty}\, 
        \dtau 
    \end{align}
    for all $s\in[0,t]$.
    Hence, the next goal is to establish a suitable bound for the $L^\infty(\R^2)$-norm of $\delx U(\tau)$. Therefore, we recall that the total energy $\H\big(f(\tau)\big) $
    is conserved, meaning that $\H(f(\tau)) = \H(\mathring f)$ for all $\tau\in [0,T_*[$.
    Using the convolution representation of $U$ (cf.~\eqref{DEF:URHO}), we obtain 
    \begin{align}
    \label{EST:KIN}
        E_\mathrm{kin}\big(f(\tau)\big)
        = \H(\mathring f) 
        + \int_{\R^2} \int_{\R^2} \ln|x-y| \, \rho(\tau,y) \, \rho(\tau,x) \dy\dx 
    \end{align} 
    for all $\tau\in[0,T_*[$. Unfortunately, as the total charge is generally not zero, we cannot directly use Lemma~\ref{lem:estimates_potential_energy} to bound second summand on the right-hand side, as it has the ``wrong'' sign. We thus expand this term as follows:
    \begin{align*}
        &\int_{\R^2} \int_{\R^2} 
        \ln|x-y| \, \rho(\tau,y) \, \rho(\tau,x) \dy\dx 
        \\[1ex]&\quad
        = \int_{\R^2} \int_{\R^2} 
        \ln|x-y| \, \rho^+(\tau,y) \, \rho^+(\tau,x) \dy\dx 
        \\ &\quad\qquad
        + \int_{\R^2} \int_{\R^2} 
        \ln|x-y| \, \rho^-(\tau,y) \, \rho^-(\tau,x) \dy\dx 
        \\ &\quad\qquad
        - 2 \int_{\R^2} \int_{\R^2} 
        \ln|x-y| \, \rho^+(\tau,y) \, \rho^-(\tau,x) \dy\dx 
    \end{align*} 
    for all $\tau\in[0,T_*[$. Let now $\ln_+\ge 0$ and $\ln_-\ge 0$ denote the positive and the negative part of the logarithm $\ln$, meaning that $\ln = \ln_+ - \ln_-\,$. Moreover, we define the function $\ov\rho:=\max\{\rho^+,\rho^-\}$. We thus obtain the estimate
    \begin{align}
    \label{EST:LN:1}
        &\int_{\R^2} \int_{\R^2} 
        \ln|x-y| \, \rho(\tau,y) \, \rho(\tau,x) \dy\dx 
        \notag\\[1ex]&\quad
        \le \int_{\R^2} \int_{\R^2} 
        \ln_+|x-y| \, \rho^+(\tau,y) \, \rho^+(\tau,x) \dy\dx 
        \notag\\ &\quad\qquad
        + \int_{\R^2} \int_{\R^2} 
        \ln_+|x-y| \, \rho^-(\tau,y) \, \rho^-(\tau,x) \dy\dx 
        \notag\\ &\quad\qquad
        + 2 \int_{\R^2} \int_{\R^2} 
        \ln_-|x-y| \, \ov\rho(\tau,y) \, \ov\rho(\tau,x) \dy\dx 
    \end{align} 
    for all $\tau\in[0,T_*[$. We next recall that $\rho^\pm(\tau,x) = 0$ for all $x\in\R^2$ with $\abs{x}>\mathbb X(\tau)$. Hence, the first two summands on the right-hand side can be bounded in the following way:
    \begin{align}
    \label{EST:LN:2}
        &\int_{\R^2} \int_{\R^2} 
        \ln_+|x-y| \, \rho^\pm(\tau,y) \, \rho^\pm(\tau,x) \dy\dx
        \notag\\[1ex] &\quad
        \le \ln\big(1+2\mathbb X(\tau)\big) \norm{\rho^\pm(\tau)}_{1}^2
        =\ln\big(1+2\mathbb X(\tau)\big) \bignorm{\mathring f^\pm}_{1}^2
    \end{align}
    for all $\tau\in[0,T_*[$. To bound the third summand on the right-hand side of \eqref{EST:LN:1}, we recall that $\ln_-|x-y| = -\ln|x-y|$ if $|x-y|<1$, and $\ln_-|x-y| = 0$ if $|x-y|\ge 1$.
    \begin{align}
    \label{EST:LN:3}
        &\int_{\R^2} \int_{\R^2} 
        \ln_-|x-y| \, \ov\rho(\tau,y) \, \ov\rho(\tau,x) \dy\dx 
        \notag\\[1ex] &\quad
        = - \iint\limits_{\{|x-y|<1\}} 
        \ln|x-y| \, \ov\rho(\tau,y) \, \ov\rho(\tau,x) \dy\dx 
        \notag\\[1ex] &\quad
        = - \int_{\R^2} \int_{\R^2}
        \ln|x-y| \, \ov\rho(\tau,y) \, \ov\rho(\tau,x) \dy\dx 
        \notag\\ &\quad\qquad
        + \iint\limits_{\{|x-y|\ge 1\}} 
        \ln|x-y| \, \ov\rho(\tau,y) \, \ov\rho(\tau,x) \dy\dx 
        \notag\\[1ex] &\quad
        \le - \int_{\R^2} \int_{\R^2}
        \ln|x-y| \, \ov\rho(\tau,y) \, \ov\rho(\tau,x) \dy\dx 
        \notag\\ &\quad\qquad
        + \int_{\R^2} \int_{\R^2}
        \ln_+|x-y| \, \ov\rho(\tau,y) \, \ov\rho(\tau,x) \dy\dx 
    \end{align}
    for all $\tau\in[0,T_*[$.
    Obviously, the second integral in the last line can be bounded by proceeding as in \eqref{EST:LN:2}. The first integral in the last line of \eqref{EST:LN:3} now has the ``correct'' sign such that Lemma~\ref{lem:estimates_potential_energy}(a) can be applied. Without loss of generality, we assume that $\mathring f^+ \ge 0$ is nontrivial, and we thus have 
    \begin{align*}
        0 < \bignorm{\mathring f^+}_1
        = \bignorm{\rho^+(\tau)}_1 
        \le \bignorm{\ov\rho(\tau)}_{1}
        \le \bignorm{\rho^+(\tau)}_{1}
            + \bignorm{\rho^-(\tau)}_{1}
        \le \bignorm{\mathring f^+}_{1}
            + \bignorm{\mathring f^-}_{1}
        \le C.
    \end{align*}
    By means of Lemma~\ref{lem:estimates_potential_energy}, we now obtain 
    \begin{align}
    \label{EST:LN:4}
        &- \int_{\R^2} \int_{\R^2}
        \ln|x-y| \, \ov\rho(\tau,y) \, \ov\rho(\tau,x) \dy\dx
        \notag\\[1ex] &\quad
        \le \frac12 \norm{\ov\rho(\tau)}_1 \int_{\R^2} \ov\rho(\tau,x)\; 
            \ln\left(
            \frac{\ov\rho(\tau,x)}{\norm{\ov\rho(\tau)}_1}
            \right) \dx + C\norm{\ov\rho(\tau)}_1^2
        \notag\\[1ex] &\quad
        \le C + C \ln\big(1 + \norm{\rho^+(\tau)}_\infty + \norm{\rho^-(\tau)}_\infty\big)
    \end{align}
    for all $\tau\in[0,T_*[$.
    Combining the estimates \eqref{EST:LN:1}--\eqref{EST:LN:4}, we can estimate the right-hand side of \eqref{EST:KIN}. This gives
    \begin{align}
    \label{EST:KIN:1}
        E_\mathrm{kin}\big(f(\tau)\big)
        \le C + C \ln\big(1+2\mathbb X(\tau)\big) 
            + C \ln\big(1+ \norm{\rho^+(\tau)}_\infty + \norm{\rho^-(\tau)}_\infty\big)
    \end{align} 
    for all $\tau\in[0,T_*[$.
    Since any solutions $(X^\pm,P^\pm)$ of the characteristic systems \eqref{CHARSYS} (written for $F=-\delx U$ and $G=B$) satisfy $\dot X^\pm = P^\pm$, it is straightforward to check that
    \begin{align*}
        \mathbb X(\tau) 
        \le \mathbb X(0) + \int_0^\tau \mathbb P(s) \ds
        \le C + \tau \mathbb P(\tau)
        \quad\text{for all $\tau\in[0,T_*[$.}
    \end{align*}
    Furthermore, we have
    \begin{align}
    \label{EST:RHO:INF}
        \norm{\rho^\pm(\tau)}_{L^\infty(\R^2)} 
        \le \frac{4\pi}{3} \bignorm{\mathring f^\pm}_\infty \mathbb P(\tau)^3
        \le C \mathbb P(\tau)^3
        \quad\text{for all $\tau\in[0,T_*[$.}
    \end{align}
    Let now $0<\delta<\frac 12$ be arbitrary. In the following, let $C_\delta$ denote generic positive constants that depend only on $\delta$, $\bignorm{\mathring f^\pm}_1$ and $\bignorm{\mathring f^\pm}_\infty$, and may change their value from line to line.
    In combination with \eqref{EST:KIN:1}, we conclude that
    \begin{align}
        \label{EST:KIN:2}
        E_\mathrm{kin}\big(f(\tau)\big)
        &\le C + C \ln\big(C(1+\tau)\big(1+\mathbb P(\tau)\big)\big) 
            + C \ln\big( C \big(1+\mathbb P(\tau)\big)^3 \big)
        \notag\\ &\quad
        \le C_\delta (1+\tau)^{2\delta} \big(1+\mathbb P(\tau)\big)^{2\delta}
    \end{align}
    for all $\tau\in[0,T_*[$.
    By means of an interpolation argument, we further obtain the estimate
    \begin{align*}
        \rho^\pm(\tau,x) 
        \le C \left( \int_{\R^3}\abs{p}^2 f^\pm(\tau,x,p) \dv \right)^{\frac 35}
    \end{align*}
    for all $\tau\in [0,T_*[$ and $x\in \R^2$.
    Along with \eqref{EST:KIN:2}, this leads to the estimate
    \begin{align}
    \label{EST:RHO:53}
        \norm{\rho(\tau)}_{L^{5/3}(\R^2)}^{5/3}
        &= \int_{\R^2} |\rho(\tau,x)|^{5/3} \dx
        \le C E_\mathrm{kin}\big(f(\tau)\big)
        \notag \\
        &\le C_\delta (1+\tau)^{2\delta} 
            \big(1+\mathbb P(\tau)\big)^{2\delta}
    \end{align}
    for all $\tau\in [0,T_*[$. Invoking the estimate \eqref{EST:DU} from Lemma~\ref{LEM:EP}, as well as the estimates \eqref{EST:RHO:INF} and \eqref{EST:RHO:53}, we thus get
    \begin{align*}
        \norm{\delx U(\tau)}_{L^\infty(\R^2)} 
        &\le C \norm{\rho(\tau)}_{L^{5/3}(\R^2)}^{5/6}
            \norm{\rho(\tau)}_{L^{\infty}(\R^2)}^{1/6} 
        \notag\\
        &\le C_\delta (1+\tau)^\delta \big( 1 + \mathbb P(\tau) \big)^{\delta+1/2}
    \end{align*}
    for all $\tau\in [0,T_*[$. Using this estimate to bound the right-hand side of \eqref{EST:PS}, we obtain 
        \begin{align*}
        \big(1 + \mathbb P(s)\big)
        \le \big(1 + \mathbb P(0)\big) 
        + \int_0^s C_\delta (1+\tau)^\delta\,\big( 1 + \mathbb P(\tau)\big)^{\delta + 1/2} \dtau 
    \end{align*}
    for all $s\in [0,t]$.
    Applying a nonlinear generalization of Gronwall's lemma (see \cite[Theorem~21]{Dragomir} with $\alpha=\frac 12 + \delta$), and recalling that $\delta < \frac 12$, we eventually conclude that
    \begin{align*}
        1 + \mathbb P(t) 
        \le \left[ 
        (1+\mathbb P(0))^{\frac 12 - \delta} 
        + C_\delta \big(\tfrac 12 - \delta\big) \int_0^t (1+s)^\delta \ds
        \right]^{\frac{2}{1-2\delta}}.
    \end{align*}
    Since $t\in[0,T_*[$ was arbitrary, this yields
    \begin{align}
    \label{EST:P:P4}
        \mathbb P(t) 
        \le C_\delta(1+t)^{\frac{2\delta + 2}{1-2\delta}} 
        = C_\delta(1+t)^{2+\frac{6\delta}{1-2\delta}} 
        \quad\text{for all $t\in [0,T_*[$.}
    \end{align}
    This contradicts \eqref{CONTCRIT} and thus, it actually holds that $T_*=\infty$. 
    For any arbitrary $\eps>0$, we choose
    \begin{align*}
        \delta := \frac{\eps}{6+2\eps},
        \quad\text{i.e.,}\quad
        \eps = \frac{6\delta}{1-2\delta},
    \end{align*}
    and thus, the estimate \eqref{EST:P} follows directly from \eqref{EST:P:P4}.
    Hence, all assertions are established and the proof is complete. 
\end{proof}

\medskip

In addition, we can show that the classical solution depends continuously on the external magnetic field, respectively. This is established by the following theorem.

\begin{theorem}
\label{THM:ContDep}
    Let $\mathring f = (\mathring f^+,\mathring f^-) $ with $\mathring f^\pm \in C^1_c(\R^5;\R_0^+)$ be any initial datum, let furthermore $A^1,A^2\in C_b([0,\infty[;C^2_b(\R^2;\R^3))$ be any given external magnetic vector potentials, and let $B^1$ and $B^2$ denote the corresponding magnetic fields.   
    Moreover, for $i\in\{1,2\}$, let $f_i=(f_i^+,f_i^-)$ denote the global-in-time classical solution of the system \eqref{eq:WholeSystem} to the initial data $f_0$ corresponding to the magnetic vector potential $A^i$.
    
    Then, for any $\gamma>4$, there exists a real number $q\in(2,\infty)$ depending only on $\gamma$ as well as a constant $c>0$ depending only on $\H\big(\mathring f\big)$, $\bignorm{\mathring f^\pm}_1$, $\bignorm{\mathring f^\pm}_\infty$ and $\gamma$
    such that for all $t\ge 0$,
    \begin{align}
    \label{EST:CD}
        \sum_\pm\norm{f_1^\pm(t)-f_2^\pm(t)}_2
        &\le a_{f_2}(t) 
        \exp\big(b_{f_2}(t)\, c(1+t)^\gamma\big)
        \norm{B^1-B^2}_{L^1(0,t;L^2(\R^2))}
    \end{align}
    where 
    \begin{align*}
        a_{f_2}(t) &:= 2\underset{\pm}{\max}\; \underset{s\in[0,t]}{\max} 
        \norm{p\times \delp f_2^\pm(s)}_{L^\infty(\R^2;L^2(\R^3))},\\ 
        b_{f_2}(t) &:= 2\underset{\pm}{\max}\; \underset{s\in[0,t]}{\max} 
        \norm{\delp f_2^\pm(s)}_{L^q(\R^2;L^2(\R^3))}.
    \end{align*}
\end{theorem}


\begin{proof}
For brevity, we write
\begin{gather*}
    \bar f^\pm\coloneqq f_1^\pm-f_2^\pm,
    \quad \bar U\coloneqq U_1-U_2,
    \quad \bar B\coloneqq B^1-B^2,\\
    \bar\rho^\pm:=\rho_1^\pm-\rho_2^\pm,
    \quad \bar\rho:=\rho_1-\rho_2=\bar\rho^{\,+}-\bar\rho^{\,-}.
\end{gather*}
Subtracting the two equations
\begin{align*}
    \delt f_1^\pm + p\cdot\delx f_1^\pm \pm (-\delx U_1+p\times B^1)\cdot\delp f_1^\pm &=0,\\
    \delt f_2^\pm + p\cdot\delx f_2^\pm \pm (-\delx U_2+p\times B^2)\cdot\delp f_2^\pm &=0,
\end{align*}
we obtain
\begin{align}\label{eq:bar_f}
    \delt \bar f^\pm + p\cdot\delx \bar f^\pm \pm (-\delx U_1+p\times B^1)\cdot\delp \bar f^\pm = \mp (-\delx \bar U+p\times \bar B)\cdot\delp f_2^\pm.
\end{align}
To estimate the $L^2$-norm of $\bar f^\pm$, we multiply \eqref{eq:bar_f} by $\bar f^\pm$ and integrate with respect to $(x,p)$. This yields
\begin{align*}
    \frac12\frac{\d}{\d t}\norm{\bar f^\pm(t)}_2^2&=-\frac12\int_{\R^2}\int_{\R^3}\Big(\divg_x\big(p(\bar f^\pm)^2\big)\pm\divg_p\big((-\delx U_1+p\times B^1)(\bar f^\pm)^2\big)\Big)\dv\dx\\
    &\phantom{=\;}\mp\int_{\R^2}\int_{\R^3}\bar f^\pm(-\delx \bar U+p\times \bar B)\cdot\delp f_2^\pm\dv\dx\\
    &=\mp\int_{\R^2}\int_{\R^3}\bar f^\pm(-\delx \bar U+p\times \bar B)\cdot\delp f_2^\pm\dv\dx
\end{align*}
for all $t\ge 0$. 
In the following, let $\gamma>4$ be arbitrary, and let $c>0$ denote a generic constant that may depend on $\H\big(\mathring f\big)$, $\bignorm{\mathring f^\pm}_1$, $\bignorm{\mathring f^\pm}_\infty$ and $\gamma$ and may change its value from line to line. 
We further choose any real numbers $r\in(2,\infty)$, $q\in (1,2)$ and $\eps>0$ (all depending on $\gamma$) such that
\begin{align}
\label{RQE}
   \frac 1r + \frac 1q = \frac 12
   \quad\text{and}\quad
   \frac{6+2\eps}r+\frac{3\eps}2 + 4 < \gamma.
\end{align}
Since $\bar f^\pm(0)=0$, we have
\begin{align*}
    \frac12\norm{\bar f^\pm(t)}_2^2&=\mp\int_0^t\int_{\R^2}\int_{\R^3}\bar f^\pm(-\delx \bar U+p\times \bar B)\cdot\delp f_2^\pm\dv\dx\ds\\
    &\le \int_0^t \frac12 \big(b_{f_2}(t) \norm{\delx\bar U(s)}_r 
        + a_{f_2}(t) \norm{\bar B(s)}_2\big)\norm{\bar f^\pm(s)}_2\ds
\end{align*}
for all $t\ge 0$. By the quadratic Gronwall lemma (see \cite[Theorem~5]{Dragomir}), it follows that
\begin{align}\label{eq:barf_est_first}
    \norm{\bar f^\pm(t)}_2\le \frac12 b_{f_2}(t) \int_0^t \norm{\delx\bar U(s)}_r \ds
    + \frac12 a_{f_2}(t) \int_0^t \norm{\bar B(s)}_2\ds.
\end{align}
We now have a closer look at the term $\norm{\delx\bar U(t)}_r$. Using the Hardy--Littlewood--Sobolev inequality, we obtain the estimate
\begin{align}\label{eq:barHLS}
    \norm{\delx\bar U(t)}_r \le c\norm{\bar\rho(t)}_{r_*}
    \quad\text{for all $t\ge 0$},
\end{align}
where $r_*\in (1,2)$ is chosen such that $\frac1{r_*}=\frac12+\frac1r$. For any $t\ge 0$, we now write
\begin{align*}
    \bar {\mathbb P}^\pm(t) &:= \sup \big\{ \abs{p} \;\big|\; 
    \exists s\in[0,t]:\; (x,p) \in \supp{\bar f^\pm(s)} \big\},\\[1ex]
    \bar {\mathbb X}^\pm(t) &:= \sup \big\{ \abs{x} \;\big|\; 
    \exists s\in[0,t]:\; (x,p) \in \supp{\bar f^\pm(s)} \big\},\\[1ex]
    \bar {\mathbb P}(t) &:= \underset{\pm}{\max}\; \bar {\mathbb P}^\pm(t),
    \quad \bar {\mathbb X}(t) := \underset{\pm}{\max}\; \bar {\mathbb X}^\pm(t).
\end{align*}
As a consequence of Theorem \ref{thm:globalWP}, it holds that
\begin{align*}
    \bar{\mathbb P}(t)\le c (1+t)^{2+\eps}
    \quad\text{and}\quad
    \bar{\mathbb X}(t)\le c (1+t)^{3+\eps}
\end{align*}
for all $t\ge 0$.
Since
\[
    \abs{\bar\rho(t,x)}=\abs{\int_{\R^3}(\bar f^+-\bar f^-)(t,x,p)\dv}\le \sqrt{4\pi/3}\; 
    \bar{\mathbb P}(t)^{3/2} \left(\int_{\R^3}\abs{(\bar f^+-\bar f^-)(t,x,p)}^2\dv\right)^{1/2},
\]
we thus infer that
\begin{align*}
    \norm{\bar\rho(t)}_{r_*}&\le\pi^{1/r}\bar{\mathbb X}(t)^{2/r}\norm{\bar\rho(t)}_2
    \le \sqrt{8 \pi/3} \;\pi^{1/r}\,\bar{\mathbb X}(t)^{2/r}\bar{\mathbb P}(t)^{3/2}\sum_\pm\norm{\bar f^\pm(t)}_2\\
    &\le c (1+t)^{\frac{6+2\eps}r+\frac{3\eps}2+3}\sum_\pm\norm{\bar f^\pm(t)}_2
\end{align*}
for all $t\ge 0$.
In combination with \eqref{eq:barHLS} and \eqref{eq:barf_est_first}, we get
\[\sum_\pm\norm{\bar f^\pm(t)}_2\le a_{f_2}(t) \norm{\bar B}_{L^1(0,t;L^2(\R^2))}+ b_{f_2}(t)\, c \int_0^t (1+s)^{\frac{6+2\eps}r+\frac{3\eps}2+3}\sum_\pm\norm{\bar f^\pm(s)}_2\ds,\]
and thus, Gronwall's lemma implies that
\[
    \sum_\pm\norm{\bar f^\pm(t)}_2\le a_{f_2}(t)\norm{\bar B}_{L^1(0,t;L^2(\R^2))}
    \exp\left(b_{f_2}(t)\, c(1+t)^{\frac{6+2\eps}r+\frac{3\eps}2+4}\right)
\]
for all $t\ge 0$. Due to \eqref{RQE}, this proves \eqref{EST:CD} and thus, the proof is complete. 
\end{proof}


%
%

\section{Existence of confined steady states}\label{sec:Existence}
 
We now show the existence of a confined steady state of the system \eqref{eq:WholeSystem} under the influence of an external magnetic field $B=\curl_x A$. Such steady states are defined as follows:

\begin{definition}[Steady states] \label{DEF:SS}
    A pair of functions $f_0=(f_0^+,f_0^-)$ is called a \emph{steady state} of the Vlasov--Poisson system \eqref{eq:WholeSystem} if $f_0^\pm$, the induced electric potential $U_0$, the external vector potential $A^0$ and its induced magnetic field $B^0$ have the regularity
    \begin{align*}
        f_0^\pm\in C^1(\R^5;\R_0^+), \quad
        U_0\in C^2(\R^2), \quad
        A^0\in C^2(\R^2;\R^3), \quad
        B^0 \in C^1(\R^2;\R^3),
    \end{align*}
    and satisfy the following equations:
\begin{subequations}
	\begin{alignat}{2}
	&p\cdot\delx f_0^\pm \pm(-\delx U_0+p\times B^0)\cdot\delp f_0^\pm =0 &&\quad\mathrm{on}\ \R^2\times\R^3,
	\\[1ex]
	&-\Delta_x U_0=4\pi\rho_0
	&&\quad\mathrm{on}\ \R^2,\label{eq:Poisson:SS}
	\\[2ex]
	& \lim_{|x|\to\infty}(U_0(x)+2M\ln|x|)=0, 
	\label{eq:PoissonBC:SS} \\[1ex]
	&\rho_0^\pm = \int_{\R^3} f_0^\pm(\cdot,\cdot,p) \, \mathrm dp, 
	    \quad \rho_0 = \rho_0^+ - \rho_0^-
	&&\quad\mathrm{on}\ \R^2,
	\\[1ex]
	&\curl_x A^0 = B^0
	&&\quad\mathrm{on}\ \R^2,
	\end{alignat}
\end{subequations}
where $M=\int_{\R^2}\rho_0\,dx$. A steady state $f_0^\pm$ is called \emph{confined in a cylinder with radius $R>0$} if $f_0^\pm(x,p)=0$ for all $\abs{x}>R$, and it is called \emph{nontrivial} if both $f_0^\pm$ are not identically zero.
\end{definition}

We are looking for steady states which are axially symmetric about the $z$-axis, meaning that
\begin{align*}
    f_0^\pm(x,p) = f_0^\pm(Rx,R\bar p,p_3), \,
    U_0(x) = U_0(Rx), \,
    A^0(x) = RA^0(Rx), \,
    B^0(x) = RB^0(Rx)
\end{align*}
for all $x\in\R^2$, $p=(\bar p, p_3)\in \R^3$, and any rotation matrix $R\in \mathrm{SO}(2)$. In the following, we will sometimes write 
\[f_0^\pm(r,p)=f_0^\pm(x,p),\;\; U_0(r)=U_0(x),\;\; A^0(r)=A^0(x),\;\; B^0(r)=B^0(x)\]
with some abuse of notation.


To construct a nontrivial, confined steady state, we make the general ansatz
\begin{align*}
    f_0^\pm=\eta^\pm(\E^\pm,\F^\pm,\G^\pm)
\end{align*}
where $\eta^\pm\in C^1(\R^3)$, and
\begin{align*}
    \E^\pm(x,p) := \frac 12 \abs{p}^2 \pm U_0, \quad
    \F^\pm(x,p) := r(p_\varphi \pm A_\varphi^0), \quad
    \G^\pm(x,p) := p_3 \pm A_3^0
\end{align*}
are invariants of the characteristic flow.

Furthermore, also due to the axial symmetry, we need some further assumptions on the external magnetic vector potential.


\paragraph{General assumption on $A^0$.}
\begin{enumerate}[label = $\mathrm{(A\arabic*)}$, ref = $\mathrm{(A\arabic*)}$, start = 0]
    \item\label{COND:A0} The external magnetic vector potential $A^0\in C^2(\R^2;\R^3)$ satisfies $A_r^0=0$ everywhere on $\R^2$ and $A_\varphi^0=A_3^0=0$ on the $z$-axis.
\end{enumerate}
To make this assumption plausible, we point out three facts: First, $A_r^0$ does not affect $B$ and should satisfy $\frac1r \partial_r(rA_r^0)=0$ in view of the gauge condition $\divg_xA^0=0$, so that the choice $A_r^0=0$ is no loss of generality. Second, $A_\varphi^0$ necessarily has to vanish on the $z$-axis in order for $B^0=\curl_xA^0$ to be well-defined on the $z$-axis in the classical sense. Third, since adding constants to $A^0$ does not affect $B^0$, the choice $A_3^0=0$ on the $z$-axis is also no loss of generality.

For any given radius $R_c>0$ of the reaction chamber and any radius $0<R<R_c$, the following theorem ensures the existence of a nontrivial steady state that is confined in the infinite cylinder about the $z$-axis with radius $R$. Note that both a $\theta$-pinch and a $z$-pinch configuration are considered.

\begin{theorem}[Existence of confined steady states] \label{THM:Existence}
    Suppose that \ref{COND:A0} holds and let $R$, $R_c\in\R$ with $0<R<R_c$ be arbitrary.
    We further assume that the ansatz functions $\eta^\pm \in C^1(\R^3)$ and the axially symmetric external magnetic vector potential $A^0\in C^2(\R^2;\R^3)$ satisfy the following assumptions:
    \begin{enumerate}[label = $\mathrm{(A\arabic*)}$, ref = $\mathrm{(A\arabic*)}$]
    \item\label{COND:A1}
    There exist functions $\eta_*^\pm\in L^1(\R^2)$ such that 
    \begin{align*}
    	\eta^\pm(\E,\F,\G)\le\eta_*^\pm(\E,\G) \quad\text{for all $(\E,\F,\G)\in\R^3$.}
    \end{align*}
    \item\label{COND:A2} 
    There exist functions $\eta_\#^\pm\in L^1(\R)$ such that 
    \begin{align*}
    	|\del_{\E}\eta^\pm(\E,\F,\G)|\le\eta_\#^\pm(\E,\G) \quad\text{for all $(\E,\F,\G)\in\R^3$,}
    \end{align*}
    \item\label{COND:A3}
    There exist cut-off energies $\E_{\max}^\pm > 0$ such that
    \begin{align*}
    	\eta^\pm(\E,\F,\G)=0 \quad\text{if $\E\ge\E_{\max}^\pm$,}
    \end{align*}
     \item\label{COND:A4}
    There exist real numbers $a>0$, $b>0$ and $c>0$ such that each of the functions $\eta^\pm$ is strictly positive on $]0,a[\;\times\;]-b,0[\;\times\;]-c,c[$ or on $]0,a[\;\times\;]0,b[\;\times\;]-c,c[$.
    \item\label{COND:A5}
    If a $\theta$-pinch configuration is considered, it holds that 
    \begin{align}
    	\label{COND:THETA}
    	\eta^\pm(\E,\F,\G)=0 \quad\text{if $\;\pm\F\ge 0$}
    \end{align} 
    and $A_\varphi^0$ satisfies
    \begin{align}\label{COND:THETA_A}
        A_\varphi^0(r)\ge \max_{\pm}\sqrt{2\E_{\max}^\pm+4\pi^2\|\eta_*^\pm\|_1r^2}\quad\text{for all }r\ge R.
    \end{align}
    If a $z$-pinch configuration is considered, there exist $\G_0^\pm\in\R$ with $\pm\G_0>0$ such that 
    \begin{align}\label{COND:Z}
    	\eta^\pm(\E,\F,\G)=0 \quad\text{if $\;\pm\G\ge\pm\G_0$,}
    \end{align} 
    and that $A_3^0$ satisfies
    \begin{align}\label{COND:Z_A}
        A_3^0(r)\ge\max_{\pm}\left(|\G_0^\pm|+\sqrt{2\E_{\max}^\pm+4\pi^2\|\eta_*^\pm\|_1r^2}\right)\quad\text{for all }r\ge R.
    \end{align}
    \end{enumerate}
    Then the ansatz
    \begin{align}
    \label{ANS:SS}
        f_0^\pm=\eta^\pm(\E^\pm,\F^\pm,\G^\pm),
    \end{align}
    defines a nontrivial steady state of the Vlasov--Poisson system \eqref{eq:WholeSystem} in the sense of Definition~\ref{DEF:SS}.
    In particular, $f_0=(f_0^+,f_0^-)$ is compactly supported in $\overline{\mathcal B_R}\times\R^3\subset\R^5$ and thus 
    confined in the cylinder with radius $R<R_c$.
\end{theorem}
Here and in the following, we denote by $\mathcal B_R$ the ball in $\R^2$ about the origin with radius $R$.


For the construction of such a steady state, we follow the reasoning in \cite{knopfStSt} and \cite{weberStSt}. In \cite{knopfStSt} the existence of confined steady states for the two dimensional Vlasov--Poisson system was established, whereas in \cite{weberStSt} confined steady states for the two and one-half-dimensional Vlasov--Maxwell system were constructed.
As we are investigating steady states of the two and one-half dimensional Vlasov--Poisson system, the arguments of \cite{knopfStSt} and \cite{weberStSt} can be combined.
Indeed, since our line of argument is very similar, we only outline the strategy and sketch the most important steps.

\begin{proof}[Sketch of the proof]
As in \cite{knopfStSt}, we first observe that the densities $\rho_0^\pm$ can be expressed by means of the ansatz functions $\eta^\pm$ via
\begin{align} 
    &\rho_0^\pm(r)=\int_{\R^3}f_0^\pm(r,p)\dv
    =\int_{\R^3}\eta^\pm\left(\frac12|p|^2\pm U_0(r),r(p_\varphi\pm A_\varphi^0(r)),p_3\pm A_3^0(r)\right)\dv\nonumber\\[1ex]
    &=\int_{\R}\int_0^\infty\int_0^{2\pi}u\, \eta^\pm \left(\frac12u^2+\frac12p_3^2\pm U_0(r),r(u\sin\theta\pm A_\varphi^0(r)),p_3\pm A_3^0(r)\right)\dint\theta\dint u\dv_3\nonumber\\[1ex]
    &=\int_{\R}\int_{\frac12(\G\mp A_3^0(r))^2\pm U_0(r)}^\infty\int_0^{2\pi}
    \eta^\pm\Big(\E,r\big(\sqrt{2(\E\mp U_0(r))-(\G\mp A_3^0(r))^2}\sin\theta+A_\varphi^0(r)\big),\G\Big)
    \nonumber\\[-1ex]
    &\hspace{0.8\textwidth}\dint\theta\dint\E\dint\G.\label{eq:rho_intermsof_U}
\end{align}
Note that the integral is well defined due to the assumption~\ref{COND:A1}.
Under our symmetry assumptions, the Laplacian simplifies to $\Delta_x=\frac1r\delr(r\delr)$. Thus, \eqref{eq:Poisson:SS} can be integrated with respect to $r$, which yields
\begin{align}\label{eq:U_intermsof_rho}
    U_0(r)=-4\pi\int_0^r\frac1s\int_0^s\sigma\rho_0(\sigma)\dint\sigma\dint s.
\end{align}
We point out that in this formula, the boundary condition $U_0(0)=0$ is incorporated. Thus, in general, the boundary condition \eqref{eq:PoissonBC:SS} cannot be satisfied. However, under the assumption that $\rho_0$ is compactly supported, it is easy to see that the limit in \eqref{eq:PoissonBC:SS} exists and is finite, but not necessarily vanishes. Since adding constants to the electric potential does not affect the electric field, \eqref{eq:PoissonBC:SS} can be replaced by the boundary condition $U_0(0)=0$ which corresponds to \eqref{eq:U_intermsof_rho}. For a more detailed discussion see \cite[Sect.~3]{knopfStSt}.

After merging \eqref{eq:rho_intermsof_U} and \eqref{eq:U_intermsof_rho} the problem of finding a steady state can then be regarded as a fixed point problem for $U_0$. Under the assumption~\ref{COND:A2},
on each bounded interval $[0,\delta]$, $\delta>0$, this fixed point problem (in the space of continuous functions) can be tackled by Schaefer's fixed point theorem (cf. \cite{weberStSt}) or a direct fixed point iteration (cf. \cite{knopfStSt}) to obtain a unique solution, which gains regularity a posteriori by means of \eqref{eq:U_intermsof_rho}. 
Now, that $U_0$ is constructed, the quantity $\E^\pm$ is known explicitly. Since $\E^\pm$, $\F^\pm$ and $\G^\pm$ are invariants of the characteristic flow, we conclude that the ansatz \eqref{ANS:SS} defines a steady state $f_0=(f_0^+,f_0^-)$ of the Vlasov--Poisson system \eqref{eq:WholeSystem} (with \eqref{eq:PoissonBC:SS} being replaced by the equivalent boundary condition $U_0(0)=0$) in the sense of Definition~\ref{DEF:SS}.

Having the steady state $f_0$ at hand, it remains to establish its claimed properties. Firstly, the assumption~\ref{COND:A3} suffices to ensure that $f_0$ is compactly supported with respect to the velocity variable $p$ (cf. \cite[Theorem~4.8(i)]{weberStSt}).  
Secondly, the assumption~\ref{COND:A4} implies that $f_0$ is nontrivial (cf. \cite[Theorem~4.8(ii)]{weberStSt}).
Thirdly, and most importantly, we need to show that the steady state is confined within a cylinder with radius $R<R_c$. Proceeding as in the proof of \cite[Theorem~5.1]{weberStSt} this assertion can be established by means of the assumption~\ref{COND:A5} which ensures that the magnetic field $B^0$ is sufficiently strong to confine the plasma. We point out that, in order to satisfy \eqref{COND:THETA_A} (or \eqref{COND:Z_A}), the function $A_\varphi^0$ (or $A_3^0$) cannot be chosen completely arbitrarily on $[0,R]$ due to the condition $A_\varphi^0(0)=0$ (or $A_3^0(0)=0$) resulting from the normalization at $r=0$.
Now all claims are established.
\end{proof}

\medskip

Let us make some concluding remarks.

\begin{remark}
    \begin{enumerate}[label = $\mathrm{(\alph*)}$, ref = $\mathrm{(\alph*)}$]
    \item The assumptions \ref{COND:A1}--\ref{COND:A5} entail that $a \le \min\{\E_{\max}^+,\E_{\max}^-\}$, and that $c \le \min\{\G_0^+,-\G_0^-\}$ in the $z$-pinch case.
    \item In the same fashion we could also consider the corresponding the \enquote*{sign-reversed} cases where the conditions $\pm\F\ge 0$ in \eqref{COND:THETA} or $\pm\G\ge\pm\G_0>0$ in \eqref{COND:Z_A} are replaced by $\pm\F\le 0$ or $\pm\G\le\pm\G_0<0$, respectively. In these alternative cases, the results of Theorem~\ref{THM:Existence} could be established analogously.
    \item It would certainly also be possible to construct a nontrivial, confined steady state in a \emph{screw-pinch configuration} that is a combination of a $\theta$-pinch and a $z$-pinch. In this case the external magnetic vector potential would be given as a superposition of two vector-potentials $A_\varphi^0$ and $A_3^0$, i.e., $A^0 = A_\varphi^0 + A_3^0$, such that \ref{COND:A5} is satisfied in the sense that \eqref{COND:THETA} and \eqref{COND:Z} hold, and \eqref{COND:THETA_A} or \eqref{COND:Z_A} hold for all $r\ge R$. However, for simplicity, we restrict ourselves to the investigation of the pure $\theta$-pinch and $z$-pinch configurations.
    \item It is important to point out that we slightly differ from \cite{weberStSt} as follows: In \cite{weberStSt} the ansatz \eqref{ANS:SS} is only made in the region $r\le R_c$ and outside $f_0^\pm$ is defined to be zero (that is to say, possible shells in the region $r>R_c$ are left out). This causes no problem thanks to $f_0^\pm=0$ on $[R,R_c]$. As a consequence, \eqref{COND:THETA_A} or \eqref{COND:Z_A} only need to be imposed on $[R,R_c]$, but it is unclear whether the relation \eqref{ANS:SS} is also satisfied for $r>R_c$. However, in this paper it will be important---in particular, for \eqref{eq:Case2}---that \eqref{ANS:SS} holds for all $r\ge 0$. Thus, \eqref{COND:THETA_A} or \eqref{COND:Z_A} have to be imposed on $[R,\infty[$.
    \end{enumerate}
\end{remark}
 

%
%

\section{Stability of confined steady states}\label{sec:stability}

We split the investigation of stability properties into two parts: First, the external magnetic potential is fixed and the initial data are perturbed. Second, the initial data are fixed and the external magnetic potential is perturbed. In the end we show that both results can also be combined.

Throughout this section, we will use the notation
\[\RR(h)\coloneqq\min\{r\ge0\mid h(x,v)=0\text{ if }|x|\ge r\}\]
for any axially symmetric $h\in C_c(\R^5)^j$ with $j\in\{1,2\}$.

In Subsection \ref{sec:Pert_Initdata}, we will mainly follow the ideas of \cite{BRV_Stability}. We fix an external magnetic potential $A^0$ satisfying \ref{COND:A0} and make the following ansatz for the steady state $f_0^\pm$:
\begin{align}
    \label{DEF:F0}
    f_0^\pm = \eta^\pm(\E^\pm,\F^\pm,\G^\pm) = \vartheta^\pm(\E^\pm) \, \psi^\pm(\F^\pm,\G^\pm).
\end{align}

In the following, we consider a $\theta$-pinch configuration. In order to satisfy the assumptions made in Section~\ref{sec:Existence}, we assume
\begin{enumerate}[label = $\mathrm{(S\arabic*)}$, ref = $\mathrm{(S\arabic*)}$]
    \item\label{COND:S1} $\vartheta^\pm\in C^1(\R)\cap W^{1,1}(\R)$, and there exists a $\E_{\max}^\pm>0$ such that $\vartheta^\pm(\tau)=0$ for all $\tau\ge\E_{\max}^\pm$ and $\vartheta^\pm(\tau)>0$ for all 
    $\tau < \E_{\max}^\pm$;
    \item
    $\psi^\pm\in C^1(\R^2)$, $\psi^\pm(\sigma,\mu)>0$ for $\pm \sigma<0$, and $\psi^\pm(\sigma,\mu)=0$ for $\pm \sigma\ge 0$, and there exist $\psi_\ast^\pm\in L^1(\R)$ such that $\psi^\pm(\sigma,\mu)\le\psi_\ast^\pm(\mu)$ for all $(\sigma,\mu)\in\R^2$;
    \item\label{COND:S3} there exists $0<\tilde R<R_c$ such that $A_\varphi^0(r)\ge \max_{\pm}\sqrt{2\E_{\max}^\pm+4\pi^2\|\vartheta^\pm\|_1\|\psi_*^\pm\|_1\,r^2}$ for all $r\ge\tilde R$.
\end{enumerate}
Here, $R_c>0$ stands again for the fixed radius of the cylindrical reaction chamber.

Since the assumptions \ref{COND:S1}--\ref{COND:S3} imply \ref{COND:A1}--\ref{COND:A5} (with $\tilde R$ instead of $R$ in \ref{COND:A5}), Theorem~\ref{THM:Existence} ensures that the steady state $f_0$ in \eqref{DEF:F0} actually exists, is nontrivial and confined in a cylinder with radius $R=\RR(f_0)\le\tilde R < R_c$. In particular, $R$ is thus strictly smaller than the radius $R_c$ of the reaction chamber.

We let $R^\pm\coloneqq \RR(f_0^\pm)$, $\E_0^\pm\coloneqq\frac12|p|^2\pm U_0$, and
\[\E_{\min}^\pm\coloneqq\inf_{\mathcal B_{R^\pm}\times\R^3}\E_0^\pm=\inf_{\mathcal B_{R^\pm}}\pm U_0.\]
Notice that $\E_{\min}^\pm<\E_{\max}^\pm$ since the steady state $f_0=(f_0^+,f_0^-)$ is nontrivial.

Moreover, we assume
\begin{enumerate}[label = $\mathrm{(S\arabic*)}$, ref = $\mathrm{(S\arabic*)}$, resume]
    \item\label{COND:S4} $(\vartheta^\pm)'(\tau)<0$ for $\tau\in[\E_{\min}^\pm,\E_{\max}^\pm[$,
\end{enumerate}
which is crucial in order to obtain stability.

For $f=(f^+,f^-)\in C_c(\R^5)^2$, we introduce the kinetic energy 
\[E_{\mathrm{kin}}(f)=\sum_\pm\frac12\int_{\R^5} |p|^2f^\pm\dz\]
and the potential energy
\[E_{\mathrm{pot}}(f)=\frac12\int_{\R^2}U_f\rho_f\dx=-\int_{\R^2}\int_{\R^2}\ln|x-y|\rho_f(y)\rho_f(x)\dy\dx,\]
where
\[\rho_f=\int_{\R^3}f^+\,dp-\int_{\R^3}f^-\,dp,\quad U_f=-2\ln|\cdot|\ast\rho_f.\]
It is very important to notice that we cannot use integration by parts to equivalently describe the potential energy as an integral over $|\delx U_f|^2$ as it would be the case in a three dimensional setting.
This is because, in contrast to a 3D situation with kernel $|\cdot|^{-1}$, the kernel $\ln|\cdot|$ does not vanish at infinity. 
We further recall that the total energy
\[\H(f)=E_{\mathrm{kin}}(f)+E_{\mathrm{pot}}(f)\]
is conserved along classical solutions of \eqref{eq:WholeSystem}.

\subsection{Stability with respect to perturbations of the initial data}\label{sec:Pert_Initdata}

We consider perturbations from the following set of functions:
\begin{align*}
    X&\coloneqq
    \left\{
    g=(g^+,g^-)\in C_c^1(\R^5)^2 
    \,\middle|\,
    \begin{aligned}
    & g^\pm\ge 0,\,g^\pm\text{ are axially symmetric},
    \\
    &g^\pm/(\psi^\pm\circ(\F^\pm,\G^\pm))\in L^1(\{\pm\F^\pm<0\}),\\
    &\int_{\R^2}\rho_g\,dx=\int_{\R^2}\rho_0\,dx
    \end{aligned}
    \right\}.
\end{align*}
Notice that, on the one hand, the assumption that $g$ has the same total charge as $f_0$ is reasonable since physically meaningful perturbations preserve the total charge, and, on the other hand, important for the positive definiteness of the potential energy induced by a difference $g-f_0$; cf. Lemma \ref{lem:estimates_potential_energy}(b).

Before we proceed, we first recall that the potential $U_0$ satisfies \eqref{eq:U_intermsof_rho}. We now define $M(r)\coloneqq 2\pi\int_0^rs\rho_0(s)\ds$, i.e., $M(r)$ is the charge in a sliced circle with radius $r$, and the total charge in each slice can be expressed as $M = M(R)$. Then the electric potential $U_0$ can be written as
\begin{align*}
    U_0(r)=\begin{cases}-2\int_0^r\frac{M(s)}{s}\ds,&0\le r\le R,\\-2\int_0^R\frac{M(s)}{s}\ds-2M(\ln r-\ln R),&r>R.\end{cases} 
\end{align*}
Therefore, the function 
\begin{align}
    \label{DEF:XI}
    \xi(r)\coloneqq\sup_{\mathcal B_r}|U_0|,\quad r > 0  
\end{align}
grows logarithmically on $[R,\infty[$ if $M\neq 0$ and is constant on $[R,\infty[$ if $M=0$.

\subsubsection{Construction of the energy-Casimir functional}
To prove stability with respect to perturbations of the initial data, we employ the \emph{energy-Casimir method}. 
The idea is to construct a \textit{Casimir functional}
\begin{align}\label{eq:Casimir}
    \C(f)=\sum_\pm\int_{\R^2}\int_{\R^3}\Phi^\pm(f^\pm,\F^\pm,\G^\pm)\dv\dx
\end{align}
in such a way that its derivative at $f_0$ matches exactly or at least dominates the negative of the linear part of the expansion
\begin{align}
    \H(f)&=\H(f_0)+\sum_\pm\int_{\R^2}\int_{\R^3}\E_0^\pm(f^\pm-f_0^\pm)\dv\dx\nonumber\\
    &\phantom{=\;}-\int_{\R^2}\int_{\R^2}\ln|x-y|(\rho_f-\rho_0)(y)(\rho_f-\rho_0)(x)\dy\dx.\label{eq:H_expansion}
\end{align}
Note that for any choice of $\Phi$ any $f\in X$, the expression $\mathcal C(f)$ is constant in time. 
For more background on the energy-Casimir method we refer to \cite{BRV_Stability,GuoRein,Rein_Casimir,Rein_Casimir_2,HMRW_Casimir}. 

Let $\vartheta_{\max}^\pm\coloneqq\vartheta^\pm(\E_{\min}^\pm)$. By assumption, the map $\vartheta^\pm\colon[\E_{\min}^\pm,\E_{\max}^\pm]\to[0,\vartheta_{\max}^\pm]$ is strictly decreasing and onto. 
We now construct the functions $\Phi^\pm\colon[0,\infty[\times\R\times\R\to\R$ which defines the Casimir functional $\C(f)$. 
Note that $\Phi^\pm=\Phi^\pm(\tau,\sigma,\mu)$ can be chosen arbitrarily (e.g., zero) if $\sigma=0$ since the set $\{\F^\pm=0\}$ has Lebesgue measure zero. First, we consider $\pm\sigma<0$. For $\mu\in\R$ and $\tau\in[0,\vartheta_{\max}^\pm\psi^\pm(\sigma,\mu)]$ we define
\[\Phi^\pm(\tau,\sigma,\mu)\coloneqq-\psi^\pm(\sigma,\mu)\int_0^{\tau/\psi^\pm(\sigma,\mu)}(\vartheta^\pm)^{-1}(s)\ds.\] 
Clearly, $\Phi^\pm(\cdot,\sigma,\mu) \in C^1([0,\vartheta_{\max}^\pm\psi^\pm(\sigma,\mu)])\cap C^2(]0,\vartheta_{\max}^\pm\psi^\pm(\sigma,\mu)])$ with
\[\del_\tau^\pm\Phi(\tau,\sigma,\mu)=-(\vartheta^\pm)^{-1}\left(\frac{\tau}{\psi^\pm(\sigma,\mu)}\right)\]
for $\tau\in[0,\vartheta_{\max}^\pm\psi^\pm(\sigma,\mu)]$ and
\begin{align*}
\del_\tau^2\Phi^\pm(\tau,\sigma,\mu)
&=-\frac{1}{(\vartheta^\pm)'\big((\vartheta^\pm)^{-1}(\tau/\psi^\pm(\sigma,\mu))\big)\psi^\pm(\sigma,\mu)} \\
&\ge -\frac{1} {\big(\inf_{[\E_{\min}^\pm,\E_{\max}^\pm]}(\vartheta^\pm)'\big)\psi^\pm(\sigma,\mu)}\eqqcolon\frac{c_{\vartheta^\pm}}{\psi^\pm(\sigma,\mu)}
\end{align*}
for $\tau\in]0,\vartheta_{\max}^\pm\psi^\pm(\sigma,\mu)]$; notice that $c_{\vartheta^\pm}>0$. Next, for $\tau>\vartheta_{\max}^\pm\psi^\pm(\sigma,\mu)$ we choose 
\begin{align*}
    \Phi^\pm(\tau,\sigma,\mu)
    &\coloneqq-\frac{(\tau-\vartheta_{\max}^\pm\psi^\pm(\sigma,\mu))^2}{2(\vartheta^\pm)'(\E_{\min}^\pm)\psi^\pm(\sigma,\mu)}-\E_{\min}^\pm(\tau-\vartheta_{\max}^\pm\psi^\pm(\sigma,\mu))\\
    &\phantom{\eqqcolon\;}-\psi^\pm(\sigma,\mu)\int_0^{\vartheta_{\max}^\pm}(\vartheta^\pm)^{-1}(s)\dint s
\end{align*}
to extend $\Phi^\pm$ to a function $\Phi^\pm\in C^1([0,\infty[\times\R_{\mp}\times\R)$ with $\Phi^\pm(\cdot,\sigma,\mu)\in C^2(]0,\infty[)$, and
\begin{alignat}{2}
    |\Phi^\pm(\tau,\sigma,\mu)|&\le c_\pm\left(\tau+\frac{\tau^2}{\psi^\pm(\sigma,\mu)}\right),
    &&\quad\tau\ge 0,\nonumber\\
    |\del_\tau\Phi^\pm(\tau,\sigma,\mu)|&\le c_\pm\left(1+\frac{\tau}{\psi^\pm(\sigma,\mu)}\right),
    &&\quad\tau\ge 0,\label{eq:Phi_1st_der_bound_1}\\
    \del_\tau^2\Phi^\pm(\tau,\sigma,\mu)&\ge\frac{c_{\vartheta^\pm}}{\psi^\pm(\sigma,\mu)},
    &&\quad\tau>0.
    \label{eq:Phi_2nd_der_pos_def_1}
\end{alignat}
for positive constants $c_\pm$ and $c_{\vartheta^\pm}$ that may depend on $\psi^\pm$ and $\vartheta^\pm$ but not on $(\tau,\sigma,\mu)$.

Let us now consider $\pm\sigma>0$. We fix $r_0>0$ such that
\begin{align}\label{eq:r_0}
    \min\{A_\varphi^0(r),r\}\ge\sqrt{2|U_0(r)|}\quad\text{for all }r\ge r_0.
\end{align}
Notice that there exists such an $r_0$ due to \ref{COND:S3} and the fact that $|U_0|$ grows at most logarithmically for large $r$. For all $\tau\ge 0$ and $\mu\in\R$, we define
\[\Phi^\pm(\tau,\sigma,\mu)\coloneqq \big(\pm\sigma+\xi(r_0)\big)\tau,\]
where $\xi$ is the function introduced in \eqref{DEF:XI}. Thus, $\Phi^\pm\in C^\infty([0,\infty[\times\R_\pm\times\R)$.

Now, for any $f\in X$, we define $\C(f)$ by \eqref{eq:Casimir}.
Notice that $\C(f)$ is well-defined because of
\begin{align*}
    &\int_{\R^2}\int_{\R^3}|\Phi^\pm(f^\pm,\F^\pm,\G^\pm)|\dv\dx\\
    &=\iint\limits_{\{\pm\F^\pm>0\}}|\Phi^\pm(f^\pm,\F^\pm,\G^\pm)|\dv\dx+\iint\limits_{\{\pm\F^\pm<0\}}|\Phi(f^\pm,\F^\pm,\G^\pm)|\dv\dx\\
    &\le \|\F^\pm f^\pm\|_1+(\xi(r_0)+c_\pm)\|f^\pm\|_1+c_\pm\|f^\pm\|_\infty\iint\limits_{\{\pm\F^\pm<0\}}\frac{f^\pm}{\psi^\pm(\F^\pm,\G^\pm)}\dv\dx<\infty.
\end{align*}
Since $\F^\pm$ and $\G^\pm$ are invariants of the characteristic flow associated with $B^0 = \curl_x A^0$, $\C$ is conserved along classical solutions of \eqref{eq:WholeSystem} to the external vector potential $A^0$. Hence, the same holds for the \emph{energy-Casimir functional} $\H_\C\coloneqq \H+\C$. In view of \eqref{eq:H_expansion}, we have the expansion
\begin{align}
    &\H_\C(f)-\H_\C(f_0)\nonumber\\
    &=\sum_\pm\int_{\R^2}\int_{\R^3}(\Phi^\pm(f^\pm,\F^\pm,\G^\pm)-\Phi^\pm(f_0^\pm,\F^\pm,\G^\pm)+\E_0^\pm(f^\pm-f_0^\pm))\dv\dx\nonumber\\
    &\phantom{=\;}-\int_{\R^2}\int_{\R^2}\ln|x-y|(\rho_f-\rho_0)(y)(\rho_f-\rho_0)(x)\dy\dx\label{eq:HC_expansion}.
\end{align}
In the next two subsections, we establish lower and upper estimates of the right-hand side of this expansion.

\subsubsection{Lower estimates on the expansion}
First, we infer from Lemma~\ref{lem:estimates_potential_energy}(b) that
\begin{align*}
    &\H_\C(f)-\H_\C(f_0)\\
    &\ge\sum_\pm\int_{\R^2}\int_{\R^3}(\Phi^\pm(f^\pm,\F^\pm,\G^\pm)-\Phi^\pm(f_0^\pm,\F^\pm,\G^\pm)+\E_0^\pm(f^\pm-f_0^\pm))\dv\dx\\
    &=\sum_\pm\left(\;\;\iint\limits_{\{\pm\F^\pm<0,f_0^\pm>0\}}+\iint\limits_{\{\pm\F^\pm<0,f_0^\pm=0\}}+\iint\limits_{\{\pm\F^\pm>0,f_0^\pm=0\}}\right)Q^\pm\dv\dx,
\end{align*}
where $Q^\pm$ is defined as
\[Q^\pm\coloneqq\Phi^\pm(f^\pm,\F^\pm,\G^\pm)-\Phi(f_0^\pm,\F^\pm,\G^\pm)+\E_0^\pm(f^\pm-f_0^\pm);\]
notice that $\{\F^\pm=0\}\subset\R^5$ has Lebesgue measure zero and that $f_0^\pm=0$ if $\pm\F^\pm>0$. Now, we fix $(x,p) \in\R^5$ and consider three cases corresponding to the decomposition above. For simplicity, we will suppress the argument $(x,p)$ in the following.

\textit{Case 1:} $\pm\F^\pm<0$, $f_0^\pm>0$: Then $\E_0^\pm\in[\E_{\min}^\pm,\E_{\max}^\pm[$ and
\[\E_0^\pm=(\vartheta^\pm)^{-1}(\vartheta^\pm(\E_0^\pm))=(\vartheta^\pm)^{-1}\left(\frac{f_0^\pm}{\psi^\pm(\F^\pm,\G^\pm)}\right)=-\del_\tau\Phi^\pm(f_0^\pm,\F^\pm,\G^\pm).\]
By the known regularity of $\Phi^\pm$ and \eqref{eq:Phi_2nd_der_pos_def_1}, we have
\begin{align*}
    Q^\pm&=\Phi^\pm(f^\pm,\F^\pm,\G^\pm)-\Phi^\pm(f_0^\pm,\F^\pm,\G^\pm)-\del_\tau\Phi^\pm(f_0^\pm,\F^\pm,\G^\pm)(f^\pm-f_0^\pm)\\
    &=\lim_{\varepsilon\to 0^+}(\Phi^\pm(f^\pm+\varepsilon,\F^\pm,\G^\pm)-\Phi^\pm(f_0^\pm,\F^\pm,\G^\pm)\\
    &\phantom{=\;\lim_{\varepsilon\to 0^+}}-\del_\tau\Phi(f_0^\pm,\F^\pm,\G^\pm)(f^\pm+\varepsilon-f_0^\pm))\\
    &=\lim_{\varepsilon\to 0^+}\frac12\del_\tau^2\Phi^\pm(\zeta_\varepsilon,\F^\pm,\G^\pm) (f^\pm+\varepsilon - f_0^\pm)^2 \ge\frac{c_{\vartheta^\pm}}{2\psi^\pm(\F^\pm,\G^\pm)}(f^\pm-f_0^\pm)^2,
\end{align*}
where $\zeta_\varepsilon$ lies between $f_0^\pm$ and $f^\pm+\varepsilon$.

\textit{Case 2:} $\pm\F^\pm<0$, $f_0^\pm=0$: Then
\begin{align}\label{eq:Case2}
    \E_0^\pm\ge\E_{\max}^\pm=(\vartheta^\pm)^{-1}(0)=-\del_\tau\Phi^\pm(0,\F^\pm,\G^\pm).
\end{align}
Thus, similarly as before,
\begin{align*}
    Q^\pm&\ge\Phi^\pm(f^\pm,\F^\pm,\G^\pm)-\Phi^\pm(0,\F^\pm,\G^\pm)-\del_\tau\Phi(0,\F^\pm,\G^\pm)f^\pm\\
    &=\lim_{\varepsilon\to 0^+}(\Phi^\pm(f^\pm+\varepsilon,\F^\pm,\G^\pm)-\Phi^\pm(\varepsilon,\F^\pm,\G^\pm)-\del_\tau\Phi^\pm(\varepsilon,\F^\pm,\G^\pm)f^\pm)\\
    &=\lim_{\varepsilon\to 0^+}\frac12\del_\tau^2\Phi^\pm(\zeta_\varepsilon,\F^\pm,\G^\pm)(f^\pm)^2\ge\frac{c_{\vartheta^\pm}}{2\psi^\pm(\F^\pm,\G^\pm)}(f^\pm)^2\\
    &=\frac{c_{\vartheta^\pm}}{2\psi^\pm(\F^\pm,\G^\pm)}(f^\pm-f_0^\pm)^2.
\end{align*}

\textit{Case 3:} $\pm\F^\pm>0$, $f_0^\pm=0$: We have 
\[Q^\pm=\big(\pm\F^\pm+\xi(r_0)+\tfrac12|p|^2\pm U_0 \big)f^\pm.\] 
To prove that $Q^\pm\ge 0$ it suffices to show that 
\[q^\pm\coloneqq \pm\F^\pm+\xi(r_0)+\frac12|p|^2\pm U_0\ge 0.\] 
To this end, we consider four sub-cases; notice that we already know that $\pm(p_\varphi\pm A_\varphi^0)>0$ due to $\pm\F^\pm>0$, and have in mind \eqref{eq:r_0}:

\textit{Case 3.1:} $r\le r_0$. Clearly, \[q^\pm\ge\xi(r_0)\pm U_0\ge 0.\]
\textit{Case 3.2:} $r>r_0$, $\pm p_\varphi\ge 0$: Here, it holds that \[q^\pm\ge rA_\varphi^0\pm U_0\ge 2|U_0|\pm U_0\ge 0.\]
\textit{Case 3.3:} $r>r_0$, $0<\mp p_\varphi<A_\varphi^0\le r$. Since the function $y\mapsto\frac12y^2-ry$ is monotonically decreasing on $[0,A_\varphi^0]\subset[0,r]$, we have \begin{align*}
    q^\pm&\ge rA_\varphi^0+\frac12(\mp p_\varphi)^2-r(\mp p_\varphi)\pm U_0\ge rA_\varphi^0+\frac12(A_\varphi^0)^2-rA_\varphi^0\pm U_0=\frac12(A_\varphi^0)^2\pm U_0\\
    &\ge|U_0|\pm U_0\ge 0.
\end{align*}
\textit{Case 3.4:} $r>r_0$, $0<\mp p_\varphi<A_\varphi^0$, $A_\varphi^0>r$: Since the function $y\mapsto\frac12y^2-ry$ attains its global minimum at $y=r$, we have 
\begin{align*}
    q^\pm&\ge rA_\varphi^0+\frac12(\mp p_\varphi)^2-r(\mp p_\varphi)\pm U_0\ge rA_\varphi^0+\frac12r^2-r^2\pm U_0\ge\frac12r^2\pm U_0\\
    &\ge|U_0|\pm U_0\ge 0.
\end{align*}
Thus, we always have $Q^\pm\ge 0$ in Case 3.

Combining all three cases, we obtain
\begin{align}
    \H_\C(f)-\H_\C(f_0)\ge\sum_\pm\frac{c_{\vartheta^\pm}}{2}\iint\limits_{\{\pm\F^\pm<0\}}\frac{(f^\pm-f_0^\pm)^2}{\psi^\pm(\F^\pm,\G^\pm)}\dv\dx.\label{eq:HC_low}
\end{align}

\subsubsection{Upper estimates on the expansion}
In the following we denote $\SS(h)\coloneqq\max\{\RR(h),R\}$ for any axially symmetric $h\in C_c(\R^5)^j$ ($j\in\{1,2\}$). Now we rewrite \eqref{eq:HC_expansion} as
\begin{align*}
    \H_\C(f)-\H_\C(f_0)&=\sum_\pm\left(\;\iint\limits_{\{\pm\F^\pm<0\}}Q^\pm\dv\dx+\iint\limits_{\{\pm\F^\pm>0\}}Q^\pm\dv\dx\right)\nonumber\\
    &\phantom{=\;}-\int_{\R^2}\int_{\R^2}\ln|x-y|(\rho_f-\rho_0)(y)(\rho_f-\rho_0)(x)\dy\dx,
\end{align*}
and estimate the three terms separately. In the case $\pm\F^\pm<0$ we have
\begin{align*}
    Q^\pm&\le|\Phi^\pm(f^\pm,\F^\pm,\G^\pm)-\Phi^\pm(f_0^\pm,\F^\pm,\G^\pm)|+|\E_0^\pm(f^\pm-f_0^\pm)|\\
    &\le \left(c_\pm+\frac{c_\pm\max\{\|f^\pm\|_\infty,\|f_0^\pm\|_\infty\}}{\psi^\pm(\F^\pm,\G^\pm)}+\frac12|p|^2+\xi\left(\SS(f^\pm1_{\{\pm\F^\pm<0\}})\right)\right)|f^\pm-f_0^\pm|
\end{align*}
by \eqref{eq:Phi_1st_der_bound_1}. In the case $\pm\F^\pm>0$ it holds that
\begin{align*}
    Q^\pm&\le|\Phi^\pm(f^\pm,\F^\pm,\G^\pm)|+|\E_0^\pm f^\pm|\\
    &\le\left(\pm\F^\pm+\xi(r_0)+\frac12|p|^2+\xi\left(\RR(f^\pm1_{\{\pm\F^\pm>0\}})\right)\right)f^\pm\\
    &=\left(\pm\F^\pm+\xi(r_0)+\frac12|p|^2+\xi\left(\RR(f^\pm1_{\{\pm\F^\pm>0\}})\right)\right)|f^\pm-f_0^\pm|.
\end{align*}
Lastly, making use of the assumption that $\rho_f$ and $\rho_0$ have the same total charge and applying Lemma~\ref{lem:estimates_potential_energy}(a), we find that
\begin{align*}
    &-\int_{\R^2}\int_{\R^2}\ln|x-y|\;(\rho_f-\rho_0)(y)\;(\rho_f-\rho_0)(x)\dy\dx\\
    &=-\int_{\R^2}\int_{\R^2}\ln\frac{|x-y|}{2\SS(f)}\,(\rho_f-\rho_0)(y)\,(\rho_f-\rho_0)(x)\dy\dx\\
    &\phantom{=\;}-\ln(2\SS(f))\left(\int_{\R^2}(\rho_f-\rho_0)\dx\right)^2\\
    &= -\int_{\mathcal B_{\SS(f)}}\int_{\mathcal B_{\SS(f)}}\ln\frac{|x-y|}{2\SS(f)}(\rho_f-\rho_0)(y)(\rho_f-\rho_0)(x)\dy\dx\\
    &\le-\int_{\mathcal B_{\SS(f)}}\int_{\mathcal B_{\SS(f)}}\ln\frac{|x-y|}{2\SS(f)}\; |\rho_f-\rho_0|(y) \; |\rho_f-\rho_0|(x)  \dy\dx\\
    &=-\int_{\R^2}\int_{\R^2}\ln|x-y|\;|\rho_f-\rho_0|(y)\;|\rho_f-\rho_0|(x)\dy\dx+\ln(2\SS(f))\|\rho_f-\rho_0\|_1^2\\
    &\le\frac12\|\rho_f-\rho_0\|_1\int_{\R^2}|\rho_f-\rho_0|\; \ln\frac{|\rho_f-\rho_0|}{\|\rho_f-\rho_0\|_1}\dx+(C +\ln(2\SS(f)))\|\rho_f-\rho_0\|_1^2.
\end{align*}
In summary, we have
\begin{align}
    &\H_\C(f)-\H_\C(f_0)\nonumber\\
    &\le\sum_\pm\Bigg[\;\;
    \iint\limits_{\{\pm\F^\pm<0\}}
        \Big(c_\pm+\frac{c_\pm\max\{\|f^\pm\|_\infty,\|f_0^\pm\|_\infty\}}{\psi^\pm(\F^\pm,\G^\pm)}+\frac12|p|^2 \Big) \; |f^\pm-f_0^\pm|\dv\dx\nonumber\\
    &\qquad\quad +\iint\limits_{\{\pm\F^\pm<0\}} 
        \xi\big(\SS(f^\pm1_{\{\pm\F^\pm<0\}})\big) \; |f^\pm-f_0^\pm|\dv\dx\nonumber\\
    &\qquad\quad +\iint\limits_{\{\pm\F^\pm>0\}}
        \Big(\pm\F^\pm+\xi(r_0)+\frac12|p|^2+\xi\big(\RR(f^\pm1_{\{\pm\F^\pm>0\}})\big)\Big)|f^\pm-f_0^\pm|\dv\dx\Bigg]\nonumber\\
    &\phantom{\le\;}+\frac12\|\rho_f-\rho_0\|_1\int_{\R^2}|\rho_f-\rho_0|\;\ln\frac{|\rho_f-\rho_0|}{\|\rho_f-\rho_0\|_1}\dx+(C+\ln(2\SS(f)))\|\rho_f-\rho_0\|_1^2.\label{eq:HC_up}
\end{align}

\subsubsection{The stability result}

For any initial data $\mathring f = (\mathring f^+,\mathring f^-) \in X$, we now write $f$ to denote the unique classical solution of \eqref{eq:WholeSystem} to the external vector potential $A^0$.
Recall that $\H_\C(f)$ is conserved in time. We can thus combine \eqref{eq:HC_low} and \eqref{eq:HC_up} written for $\mathring f\in X$ to obtain a stability estimate on the set $\{\pm\F^\pm<0\}$ (see \eqref{EST:STA1}). Moreover, it is obvious that all $L^q$-norms of $f^\pm-f_0^\pm$ over $\{\pm\F^\pm\ge 0\}$ are constant time, since $f_0^\pm=0$ if $\pm\F^\pm>0$ and $\{\pm\F^\pm\ge 0\}$ is invariant under the characteristic flow of the Vlasov equation with external vector potential $A^0$. Altogether, this means that the following theorem is established.

\begin{theorem}
\label{thm:Stability}
We consider the $\theta$-pinch configuration and assume that \emph{\ref{COND:S1}--\ref{COND:S4}} hold. Let $f_0$ denote the steady state introduced in \eqref{DEF:F0}. 

For any initial data $\mathring f = (\mathring f^+,\mathring f^-) \in X$ let $f=(f^+,f^-)$ denote the unique classical solution of \eqref{eq:WholeSystem}.

Then for all $q\in[1,\infty]$,
\begin{align}\label{EST:STA0}
    \|f^\pm(t)-f_0^\pm\|_{L^q(\{\pm\F^\pm\ge 0\})}=\|\mathring f^\pm-f_0^\pm\|_{L^q(\{\pm\F^\pm\ge 0\})},
\end{align}
and there exist constants $C>0$ independent of all the appearing quantities, and $c_1^\pm,c_2^\pm>0$ depending only on $\vartheta^\pm$ and $\psi^\pm$ such that for all $\mathring f \in X$ and $t\ge 0$, 
\begin{align}
    &\sum_\pm c_1^\pm\iint\limits_{\{\pm\F^\pm<0\}}\frac{(f^\pm(t)-f_0^\pm)^2}{\psi^\pm(\F^\pm,\G^\pm)}\dv\dx\nonumber\\
    &\le\sum_\pm\Bigg[\;\iint\limits_{\{\pm\F^\pm<0\}}
        \Big(c_2^\pm+\frac{c_2^\pm\max\{\|\mathring f^\pm\|_\infty,\|f_0^\pm\|_\infty\}}{\psi^\pm(\F^\pm,\G^\pm)}+\frac12|p|^2 \Big)
        \,|\mathring f^\pm-f_0^\pm|\dv\dx
    \nonumber\\
    &\qquad\quad +\iint\limits_{\{\pm\F^\pm<0\}}
        \xi\big(\SS(\mathring f^\pm1_{\{\pm\F^\pm<0\}}) \big) 
        \,|\mathring f^\pm-f_0^\pm|\dv\dx
    \nonumber\\
    &\qquad\quad +\iint\limits_{\{\pm\F^\pm>0\}}
        \Big(\pm\F^\pm+\xi(r_0)+\frac12|p|^2+\xi\big(\RR(\mathring f^\pm1_{\{\pm\F^\pm>0\}})\big)\Big)|\mathring f^\pm-f_0^\pm|\dv\dx\Bigg]\nonumber\\
    &\phantom{\le\;}+\frac12\|\rho_{\mathring f}-\rho_0\|_1\int_{\R^2}|\rho_{\mathring f}-\rho_0|\;\ln\frac{|\rho_{\mathring f}-\rho_0|}{\|\rho_{\mathring f}-\rho_0\|_1}\dx+\big(C+\ln(2\SS(\mathring f))\big)\|\rho_{\mathring f}-\rho_0\|_1^2\label{EST:STA1}.
\end{align}
\end{theorem}

\bigskip

\begin{remark}\label{rem:STA}
If $\psi^\pm$ is additionally assumed to be bounded, then from \eqref{EST:STA0} and \eqref{EST:STA1} the (slightly weaker) estimate \newpage
\begin{align*}
    &\sum_\pm \norm{f^\pm(t)-f_0^\pm}_{L^2(\R^2\times\R^3)}^2\nonumber\\
    &\le c^\pm\Bigg\{\sum_\pm\Bigg[\;\iint\limits_{\{\pm\F^\pm<0\}}
        \Big(1+\frac{\max\{\|\mathring f^\pm\|_\infty,\|f_0^\pm\|_\infty\}}{\psi^\pm(\F^\pm,\G^\pm)}+\frac12|p|^2 \Big)
        \,|\mathring f^\pm-f_0^\pm|\dv\dx
    \nonumber\\
    &\qquad\quad +\iint\limits_{\{\pm\F^\pm<0\}}
        \xi\big(\SS(\mathring f^\pm1_{\{\pm\F^\pm<0\}}) \big) 
        \,|\mathring f^\pm-f_0^\pm|\dv\dx
    \nonumber\\
    &\qquad\quad +\iint\limits_{\{\pm\F^\pm>0\}}
        \Bigg(\pm\F^\pm+\xi(r_0)+\frac12|p|^2+\xi\big(\RR(\mathring f^\pm1_{\{\pm\F^\pm>0\}})\big)\Bigg)|\mathring f^\pm-f_0^\pm|\dv\dx\Bigg]\nonumber\\
    &\phantom{\le\;}+\frac12\|\rho_{\mathring f}-\rho_0\|_1\int_{\R^2}|\rho_{\mathring f}-\rho_0|\;\ln\frac{|\rho_{\mathring f}-\rho_0|}{\|\rho_{\mathring f}-\rho_0\|_1}\dx+\big(C+\ln(2\SS(\mathring f))\big)\|\rho_{\mathring f}-\rho_0\|_1^2\Bigg\}\nonumber\\
    &\phantom{\le\;}+\|\mathring f^\pm-f_0^\pm\|_{L^2(\{\pm\F^\pm\ge 0\})}^2
\end{align*}
follows, with constants $c^\pm>0$ depending only on $\vartheta^\pm$ and $\psi^\pm$.
\end{remark}

\subsection{Stability with respect to perturbations of the magnetic vector\\ potential}

The continuous dependence estimate presented in Theorem~\ref{THM:ContDep} can be used to prove stability of the confined steady state with respect to perturbations of the external magnetic field.

\begin{theorem}
\label{thm:Stability_B}
Let $f_0=(f_0^+,f_0^-)$ denote an arbitrary steady state (in the sense of Definition~\ref{DEF:SS}) with external magnetic vector potential $A^0$ and the associate magnetic field $B^0 = \curl_x A^0$.

Moreover, let $A\in C^2(\R^2;\R^3)$ be an arbitrary (not necessarily axially symmetric) external magnetic vector potential, let $B= \curl_x A$ denote the associated magnetic field,
and let $f$ denote the corresponding classical solution of the Vlasov--Poisson system \eqref{eq:WholeSystem} to the the initial condition $f^\pm(0) = f_0^\pm$ on $\R^2\times\R^3$.

Then, for any $\gamma>4$, there exists a real number $q\in(2,\infty)$ depending only on $\gamma$ as well as a constant $c>0$ depending only on $\H\big(f_0\big)$, $\norm{f_0^\pm}_1$, $\norm{f_0^\pm}_\infty$ and $\gamma$
such that for all $t\ge 0$,
\begin{align*}
        \sum_\pm\norm{f^\pm(t)-f_0^\pm}_2
        &\le \alpha_{f_0} 
        \exp\big(\beta_{f_0}\, c(1+t)^\gamma\big)
        \norm{B-B^0}_{L^1(0,t;L^2(\R^2))}
    \end{align*}
    where 
    \begin{align*}
        \alpha_{f_0} := 2\underset{\pm}{\max}
        \norm{p\times \delp f_0^\pm}_{L^\infty(\R^2;L^2(\R^3))}, 
        \quad
        \beta_{f_0} := 2\underset{\pm}{\max}
        \norm{\delp f_0^\pm}_{L^q(\R^2;L^2(\R^3))}.
\end{align*}
\end{theorem}

\begin{proof}
The assertion follows directly from Theorem~\ref{THM:ContDep} and the fact that $f_0$ is constant in time.
\end{proof}

\subsection{Stability with respect to perturbations of both the initial data \\ and the magnetic vector potential}


We can now combine our results to obtain the following theorem.

\begin{theorem}
\label{thm:Stability_C}
We consider the $\theta$-pinch configuration and assume that \emph{\ref{COND:S1}--\ref{COND:S4}} hold and that $\psi^\pm$ is bounded.
Let $f_0=(f_0^+,f_0^-)$ denote the steady state that was introduced in \eqref{DEF:F0}. Recall that $f_0$ corresponds to the axially symmetric external magnetic vector potential $A^0\in C^2(\R^2;\R^3)$ and the associated magnetic field $B^0 =\curl_xA^0$. 

Moreover, let $A\in C^2(\R^2;\R^3)$ be an arbitrary (not necessarily axially symmetric) external magnetic vector potential, and let $B=\curl_xA $ denote the associated magnetic field.
For any given initial data $\mathring f = (\mathring f^+,\mathring f^-) \in X$, let $f=(f^+,f^-)$ and $f_*=(f_*^+,f_*^-)$ denote the corresponding classical solutions of \eqref{eq:WholeSystem} to the external magnetic potentials $A$ and $A^0$, respectively. 

Furthermore, let $\gamma>4$ be arbitrary. Then there exist
\begin{itemize}
    \item a constant $C>0$ independent of all the appearing quantities,
    \item constants $c^\pm>0$ depending only on $\vartheta^\pm$ and $\psi^\pm$,
    \item a real number $q\in(2,\infty)$ depending only on $\gamma$,
    \item a constant $c_*>0$ depending only on $\H\big(\mathring f\big)$, $\bignorm{\mathring f^\pm}_1$, $\bignorm{\mathring f^\pm}_\infty$ and $\gamma$,
\end{itemize}
such that for all $t\ge 0$,
\begin{align*}
        &\sum_\pm\norm{f^\pm(t)-f_0^\pm}_2\nonumber\\
        &\le \Bigg\{c^\pm\Bigg[\sum_\pm\Bigg(\;\iint\limits_{\{\pm\F^\pm<0\}}
        \Big(1+\frac{\max\{\|\mathring f^\pm\|_\infty,\|f_0^\pm\|_\infty\}}{\psi^\pm(\F^\pm,\G^\pm)}+\frac12|p|^2 \Big)
        \,|\mathring f^\pm-f_0^\pm|\dv\dx
    \nonumber\\
    &\qquad\quad +\iint\limits_{\{\pm\F^\pm<0\}}
        \xi\big(\SS(\mathring f^\pm1_{\{\pm\F^\pm<0\}}) \big) 
        \,|\mathring f^\pm-f_0^\pm|\dv\dx
    \nonumber\\
    &\qquad\quad +\iint\limits_{\{\pm\F^\pm>0\}}
        \Big(\pm\F^\pm+\xi(r_0)+\frac12|p|^2+\xi\big(\RR(\mathring f^\pm1_{\{\pm\F^\pm>0\}})\big)\Big)|\mathring f^\pm-f_0^\pm|\dv\dx\Bigg)\nonumber\\
    &\phantom{\le\;}+\frac12\|\rho_{\mathring f}-\rho_0\|_1\int_{\R^2}|\rho_{\mathring f}-\rho_0|\ln\frac{|\rho_{\mathring f}-\rho_0|}{\|\rho_{\mathring f}-\rho_0\|_1}\dx+\big(C+\ln(2\SS(\mathring f))\big)\|\rho_{\mathring f}-\rho_0\|_1^2\Bigg]\nonumber\\
    &\phantom{\le\;}+\|\mathring f^\pm-f_0^\pm\|_{L^2(\{\pm\F^\pm\ge 0\})}^2\Bigg\}^{1/2}\nonumber\\
        &\phantom{\le\;}+\alpha_{f_*}(t)
        \exp\big(\beta_{f_*}(t)\, c_*(1+t)^\gamma\big)
        \norm{B-B^0}_{L^1(0,t;L^2(\R^2))},
\end{align*}
where 
\begin{align*}
    \F^\pm&:=r(p_\varphi\pm A_\varphi^0),\\
    \G^\pm&:=p_3\pm A_3^0,\\
    a_{f_*}(t) &:= 2\underset{\pm}{\max}\; \underset{s\in[0,t]}{\max} 
    \norm{p\times \delp f_*^\pm(s)}_{L^\infty(\R^2;L^2(\R^3))},\\ 
    b_{f_*}(t) &:= 2\underset{\pm}{\max}\; \underset{s\in[0,t]}{\max} 
    \norm{\delp f_*^\pm(s)}_{L^q(\R^2;L^2(\R^3))}.
\end{align*}
\end{theorem}
\begin{proof}
The assertion follows directly from the decomposition $f-f_0=(f-f_*)+(f_*-f_0)$ and using the estimates of Theorem \ref{THM:ContDep} for $f-f_*$ and of Remark \ref{rem:STA} for $f_*-f_0$.
\end{proof}

\begin{remark} \normalfont
Unfortunately, in Theorem~\ref{thm:Stability_C}, the constant $c_*$ depends on the perturbed initial data $\mathring f$, and the functions $a_{f_*}$ and $b_{f_*}$ depend on the solution $f_*$ which corresponds to the perturbed initial data $\mathring f$. However, under suitable assumptions, it is possible to show that these quantities depend only on $f_0$, $A^0$ and $\gamma$.
\begin{enumerate}[label = $\textnormal{(\alph*)}$, leftmargin = *]
    \item Suppose that
\begin{align}
    \label{ASS:DELTA}
    \big|\H\big(\mathring f\big) - \H\big(f_0\big)\big| 
    \le \delta,
    \quad
    \bignorm{\mathring f^\pm - f_0^\pm}_1
    \le \delta,
    \quad
    \bignorm{\mathring f^\pm - f_0^\pm}_\infty
    \le \delta
\end{align}
for a prescribed real number $\delta>0$. Then, because of
\begin{align*}
    \big|\H\big(\mathring f\big)\big| \le \big|\H\big(f_0\big)\big| 
    + \delta,
    \quad
    \bignorm{\mathring f^\pm}_1 \le \bignorm{f_0^\pm}_1
    + \delta,
    \quad
    \bignorm{\mathring f^\pm}_\infty \le  \bignorm{f_0^\pm}_\infty
    + \delta,
\end{align*}
the number $c_*$ can be replaced by a positive constant depending only on $\H\big(f_0\big)$, $\bignorm{f_0^\pm}_1$, $\bignorm{f_0^\pm}_\infty$, $\gamma$ and $\delta$.
\item Proceeding as in \cite[Lemma~6]{knopf}, one can show that there exist increasing functions $\alpha,\beta\in C([0,\infty[)$ depending only on $A^0$, $\bignorm{\mathring f^\pm}_\infty$, and $\gamma$ such that
\begin{align*}
    a_{f_*}(t) \le \alpha(t)
    \quad\text{and}\quad
    b_{f_*}(t) \le \beta(t)
\end{align*}
for all $t\ge 0$. However, we remark that due to the Gronwall argument employed in \cite[Lemma~6]{knopf}, these functions $\alpha$ and $\beta$ exhibit exponential growth. Under the assumptions \eqref{ASS:DELTA}, $\alpha$ and $\beta$ could be replaced by similar functions depending only on $A^0$, $\bignorm{ f_0^\pm}_\infty$, $\gamma$ and $\delta$.
\end{enumerate}
\end{remark}

\paragraph{Acknowledgement} This project has received funding from the European Research Council (ERC) under the European Union’s Horizon 2020 research and innovation programme (grant agreement no 678698).





\footnotesize

\bibliographystyle{plain}
\bibliography{KW_SSVP}
%

\end{document}